\documentclass[imslayout]{imsart}
\usepackage{amsthm,latexsym, amsbsy, amssymb,multirow, epsfig, amsmath}
\usepackage{hyperref}
\usepackage{graphicx}
\usepackage{threeparttable}
\usepackage{algorithm}
\usepackage{algpseudocode}
\usepackage{xcolor}
\usepackage{natbib}
\startlocaldefs
\numberwithin{equation}{section}
\theoremstyle{plain}

\endlocaldefs

\def\wh{\widehat}
\def\wt{\widetilde}

\def\bb{\mbox{\boldmath$\beta$}}

\def\hb{\mbox{$\widehat\bb$}}

\def\a{{\bf a}}
\def\b{{\bf b}}
\def\A{{\bf A}}
\def\B{{\bf B}}
\def\U{{\bf U}}

\def\D{{\bf D}}

\def\L{{\bf L}}
\def\e{{\bf e}}

\def\I{{\bf I}}

\def\x{{\bf x}}

\def\Z{{\bf Z}}
\def\z{{\bf z}}
\def\u{{\bf u}}
\def\e{{\bf e}}
\def\u{{\bf u}}

\def\mR{\mathbb{R}}

\def\calS{\mathcal{S}}

\def\pr{\mbox{pr}}
\def\tr{\mbox{tr}}
\def\var{\mbox{var}}
\def\vec{\mbox{vec}}
\def\cov{\mbox{cov}}

\def\diag{\mbox{diag}}
\newcommand{\trans}{^{\mbox{\tiny{T}}}}
\def\defby{\stackrel{\mbox{\textrm{\tiny def}}}{=}}
\newcommand{\bGam}{\mbox{\boldmath $\Gamma$}}

\newcommand{\bSig}{\mbox{\boldmath $\Sigma$}}
\newcommand{\hSig}{\mbox{$\widehat{\bSig}$}}
\newcommand{\bOme}{\mbox{\boldmath $\Omega$}}
\newcommand{\bPsi}{\mbox{\boldmath $\Psi$}}
\newcommand{\hOme}{\mbox{$\widehat{\bOme}$}}

\newcommand{\hGam}{\mbox{$\widehat{\bGam}$}}
\newcommand{\bLam}{\mbox{\boldmath $\Lambda$}}
\newcommand{\hLam}{\mbox{$\widehat{\bLam}$}}

\newcommand{\E}{\mbox{E}}

\newcommand{\supp}{\mathcal{S}}
\DeclareMathOperator*{\argmin}{arg\,min} 

\def\beqr{\begin{eqnarray}}
\def\eeqr{\end{eqnarray}}
\def\beqrs{\begin{eqnarray*}}
	\def\eeqrs{\end{eqnarray*}}
\def\bep{\begin{prop}}
	\def\eep{\end{prop}}
\newtheorem{theo}{\bf Theorem}
\newtheorem{lemm}{\bf Lemma}
\newtheorem{prop}{\bf Proposition}

\begin{document}
	
	\begin{frontmatter}
		\title{Penalized Interaction Estimation for Ultrahigh Dimensional Quadratic Regression
		} \runtitle{Penalized Interaction Estimation}
		
		\begin{aug}
			\author{\fnms{Cheng} \snm{Wang}\thanksref{m1}
				\ead[label=e1]{chengwang@sjtu.edu.cn}},
			\author{\fnms{Binyan} \snm{Jiang}\thanksref{m2}
				\ead[label=e2]{by.jiang@polyu.edu.hk}}
			\and
			\author{\fnms{Liping} \snm{Zhu}\thanksref{m3}
				\ead[label=e3]{zhu.liping@ruc.edu.cn}
			}

			\affiliation{Shanghai Jiao Tong University\thanksmark{m1}, Hong Kong Polytechnic University
				\thanksmark{m2} and Renmin University of China\thanksmark{m3}}

			\address{Cheng Wang\\School of Mathematical Sciences\\
				Shanghai Jiao Tong University\\
				Shanghai 200240, China\\
				\printead{e1}}
			
			\address{Binyan Jiang\\Department of Applied Mathematics\\
			Hong Kong Polytechnic University\\
			Hong Kong\\
				\printead{e2}}
			
			\address{Liping Zhu\\Research Center for Applied Statistical Science \\
				Institute of Statistics and Big Data\\
				Renmin University of China\\
				Beijing 100872, China.\\
				\printead{e3}
			}
		\end{aug}
		
		\vskip1cm
			\begin{abstract}
			Quadratic regression goes  beyond the linear model by simultaneously including  main effects and interactions between the covariates. The problem of interaction estimation in high dimensional quadratic regression has received extensive attention in the past decade.  In this article we introduce a novel  method which allows us to estimate the main effects and interactions separately. Unlike existing methods for ultrahigh dimensional quadratic regressions,  our proposal  does not require the widely used heredity assumption. In addition, our proposed estimates have explicit formulas  and obey the invariance principle at the population level.  We  estimate the interactions of matrix form   under  penalized convex loss function. The resulting estimates are shown to be consistent even when the covariate dimension is an exponential order of the sample size. We develop an efficient ADMM algorithm to implement the penalized estimation.  This ADMM algorithm fully explores the cheap computational cost of matrix multiplication and is much more efficient than existing penalized methods such as all pairs LASSO.  We demonstrate the promising performance of our proposal through extensive numerical studies.
		\end{abstract}

		\begin{keyword}[class=AMS]
			62H20, 62H99, 62G99.
		\end{keyword}
		\begin{keyword}
			\kwd{High dimension; interaction estimation; quadratic regression;  support recovery.}
		\end{keyword}
		
	\end{frontmatter}
	
	\newpage
\section{INTRODUCTION}
In many scientific discoveries, a fundamental problem  is to understand how the features under investigation interact with each other. Interaction estimation has been shown to be very attractive in both parameter estimation and model prediction \citep{bien2013lasso, hao2016model}, especially for  data sets with complicated structures.  \cite{efron2004least} pointed out that for  Boston housing data, prediction accuracy can be significantly improved if interactions are  included in addition to all main effects.  In general, ignoring interactions by considering main effects alone may lead to an inaccurate or even a biased  estimation, resulting in poor prediction of an outcome of interest, whereas considering  interactions as well as main effects  can improve model interpretability and prediction substantially,  thus achieve a better understanding of how the outcome  depends on the predictive features \citep{fan2015innovated}. While it is important to identify interactions which may reveal real relationship between the outcome and the predictive features,  the number of  parameters scales squarely with that of the predictive features, making parameter estimation and model prediction very challenging for problems with large or even moderate dimensionality.

\subsection{Interaction Estimation, Feature Selection and Screening}
Estimating  interactions is a  challenging problem  because the number of pairwise interactions increases quadratically with the number of the covariates. In the past decade, there has been a surge of interest in interaction estimation in  quadratic regression. Roughly speaking, existing procedures for interaction estimation can be classified into three categories. In the first category of  low or moderate dimensional setting, standard techniques such as ordinary least squares can be readily used to estimate all the pairwise interactions as well as the main effects.  This simple one-stage strategy, however, becomes impractical or even infeasible for moderate or high dimensional problems, owing to rapid increase in dimensionality incurred by interactions. In the second category of moderate or high dimensional setting where feature selection becomes imperative, several one-stage regularization methods are proposed and some require  either the strong or the weak heredity assumption. See, for example, 
\cite{yuan2009structured}, 
\cite{choi2010variable}, \cite{bien2013lasso}, \cite{lim2015learning},  and \cite{haris2016convex}. 
These regularization methods are computationally feasible  and the theoretical properties of the resulting estimates are well understood for moderate or high dimensional problems. However, in the third category of ultrahigh dimension problems, these regularization methods are no longer feasible because their implementation requires storing and manipulating large scale design matrix and solving complex constrained optimization problems. The memory and computational cost is usually extremely expensive and prohibitive.  Very recently, several two-stage approaches are proposed for both ultrahigh dimensional regression and classification problems, including \cite{hao2014interaction},  \cite{fan2015innovated}, \cite{hao2016model} and \cite{kong2017interaction}.  Two-stage approaches estimate main effects and interactions at two separate stages, so their computational complexity is dramatically reduced. However, 
these two-stage approaches hinge heavily on either the  strong or weak heredity assumption.  These methods are computationally scalable but may completely break down when the heredity assumption is violated. 

\subsection{Heredity Assumption and Invariance Principle in Quadratic Regression}
As an extra layer of flexibility to  linear  models,   quadratic regressions include both main effects and pairwise interactions between the covariates. Denote $Y$ the outcome variable and $\x = (X_1,\ldots,X_p)\trans\in\mR^p$ the   covariate vector. For notational clarity, we define $\u \defby \E(\x)\in\mR^p$. In general, quadratic regression has the form of
\beqr\label{model}
\E(Y\mid\x) = \alpha + (\x-\u) \trans \bb + (\x-\u) \trans \bOme (\x-\u),
\eeqr
where $\alpha \in \mR^1$, $\bb = (\beta_1,\ldots,\beta_p)\trans \in\mR^p$ and $\bOme = (\bOme_{k,l})_{p\times p}  \in\mR^{p \times p}$ are all unknown parameters. 
To ensure model  identifiability, we further assume that $\bOme$ is symmetric, that is, $\bOme \trans=\bOme$, or equivalently, $\bOme_{k,l} = \bOme_{l,k}$, $1\le k,l\le p$. Our goal is to estimate $\bb$ and $\bOme$ which  characterize respectively  main effects and interactions. We remark here that the intercept $\alpha$ is also useful for  prediction.

In the literature,   heredity structures \citep{nelder1977reformulation, hamada1992analysis} have been widely imposed to   avoid  quadratic computational cost of searching over all pairs of interactions.  The heredity structures assume  that the support of $ \bOme$ could be inferred from the support of $\bb$. The strong heredity assumption requires that an interaction between two covariates be included in the model only if both main effects are important, while the weak one relaxes such a constraint to the presence of at least one main effect being important. 
In symbols, the strong and weak heredity structures are defined, respectively, as follows:
\beqrs \label{cqr}
& \textrm{strong heredity: }& \bOme_{k,l} \neq 0 \Rightarrow \beta_k^2>0  ~\textrm{\textrm{and}}~ \beta_l^2 > 0,\\
& \textrm{weak heredity: }& \bOme_{k,l} \neq 0 \Rightarrow \beta_k^2 + \beta_l^2 > 0.
\eeqrs
With the heredity assumptions, one can first seek a small number of important main effects and then only consider interactions involving these discovered main effects. It is however quite possible that main effects corresponding to important interactions are hard to detect.  An example is $Y = (1 + X_1)(1 +X_2)  + \varepsilon$,  where $X_1$ and $X_2$ are drawn independently from ${\cal N}(-1,1)$ and $\varepsilon$ is standard normal. In this example, $\cov(X_1,Y) = \cov(X_2,Y) = 0$. The main effects $X_1$ and $X_2$ are thus unlikely detectable through a working linear model $Y = \alpha_0 +  \alpha_1X_1 + \alpha_2 X_2 + \epsilon$, indicating that the heredity assumptions do not facilitate to find interactions by searching for main effects first. From a practical perspective, \cite{ritchie2001multifactor} provided a real data example to demonstrate the existence of pure interaction models  in practice. \cite{cordell2009detecting} also raised serious concerns that many existing methods that depend  on the heredity assumption may miss pure interactions in the absence of main effects. 

An ideal quantification of importance of the main effects and interactions should satisfy the invariance principle with respect to location-scale transformation of the covariates. It is natural and a common strategy to quantify the importance of main effects and interactions through the supports of $\bb$ and $\bOme$ in model (\ref{model}). In conventional linear model where only main effects are present and   interactions are absent (i.e., $\bOme = {\bf0}_{p\times p}$ in model (\ref{model})), the  invariance principle  is satisfied. In contrast, in quadratic regression (\ref{model}) with a general $\bOme$ the  invariance principle is very likely  violated. To demonstrate this issue, we can recast model (\ref{model}) as  \beqr\label{model1}
\E(Y\mid \x) = (\alpha - \u\trans\bb + \u\trans\bOme\u) + \x\trans(\bb-2\bOme\u) + \x\trans\bOme\x.\eeqr In this model, the   importance of main effects and interactions is naturally characterized through the support of $(\bb-2\bOme\u) $ and $\bOme$, respectively, indicating that  the interactions  are invariant whereas the main effects are sensitive to location transformation. In ultrahigh dimensional quadratic regression, using one-stage approaches which simultaneously estimate  main effects and interactions  under the heredity assumption or using two-stage approaches which search for main effects prior to searching for interactions in model (\ref{model}) and model (\ref{model1}) may lead to quite different conclusions. It is thus desirable to estimate interactions directly without knowing the main effects in advance. Direct interaction estimation without heredity constraints is, however, to the best of our knowledge, much more challenging and still unsolved in the literature.

\subsection{Our Contributions}
In this article we consider interaction estimation in ultrahigh dimensional quadratic regressions without heredity assumption. We make at least the following two important contributions to the literature. 
\begin{enumerate}
	\item 
	We motivate our proposal with the goal of obtaining a general and explicit expression for quadratic regression with as minimal assumptions as possible. Surprisingly, it turns out that such an explicit solution only relies on certain moment conditions on the ultrahigh dimensional covariates, which will be automatically satisfied by the widely used normality assumption. Explicit forms can be derived for both the main effects and the interactions, from which it can be seen  that the quadratic regression could be implemented as two independent tasks relating to the main effects and interactions separately. Under weaker moment assumptions, our approach is still valid in detecting the direction of the true interactions. Our proposal is different from existing one-step or two-step procedures in that we do not require the heredity assumption and our proposal give explicit forms for both the main effects and the interactions. Estimating the main effects through a separate working linear model ensures that  the resulting estimate satisfies  the  desirable invariance principle. What is more, we show that our approach for interaction detection is robust to the estimation of main effects in that even when the linear effect can not be well estimated, we can still successfully detect the interactions. 
	\item
	We show that the interaction inference is equivalent to a particular matrix estimation at the population level. We  estimate the interactions of matrix form   under  penalized convex loss function, which yields a sparse solution. We derive the theoretical consistence of our proposed  estimation when the covariate dimension is an exponential order of the sample size. Compared with the conventional penalized least squares approach, the penalization of matrix form  is appealing in both memory storage and computation cost. An efficient  ADMM algorithm is developed to implement our procedure. This algorithm fully explores the cheap computational cost for   matrix multiplication   and  is even much more efficient than existing penalized methods. We have also developed an R package ``PIE" to implement our proposal. 
\end{enumerate}

The remainder of this paper is organized as follows. We begin in Section 2 with the quadratic regression model and  derive  closed forms for both the main effects and the interactions. We propose a direct penalized estimation for high dimensional sparse quadratic model. To implement our proposal an efficient ADMM algorithm is provided. We also study the theoretical properties of our proposed estimates. We illustrate  the   performance of our proposal through simulations  in Section 3  and an application to a real world problem in Section 4.   We give some brief comments   in Section 5. All technical details are relegated to Appendix.

\section{THE ESTIMATION PROCEDURE}

\subsection{The Rationale}
In this section we discuss how to estimate $\bb$ and $\bOme$, which characterize the main effects and interactions in model (\ref{model}), respectively. Note that $\bb =E\left\{\partial E(Y\mid\x)/(\partial\x)\right\}$ and $\bOme= E\left\{\partial^2E(Y\mid\x)/(\partial\x\partial\x\trans)\right\}\big/ 2.$
Therefore, estimating $\bb$ and $\bOme$ amounts to estimating $E\left\{\partial E(Y\mid\x)/(\partial\x)\right\}$ and $E\left\{\partial^2E(Y\mid\x)/(\partial\x\partial\x\trans)\right\}$, respectively, which is however not straightforward, especially when $\x$ is ultrahigh dimensional. To illustrate the rationale of our proposal, we assume for now that   $\x$ follows  ${\cal N}(\u,\bSig)$. It follows immediately from Stein's Lemma \citep{stein1981estimation,li1992principal} that
\beqrs
&E\left\{\partial E(Y\mid\x)/(\partial\x)\right\} = \bSig^{-1}\cov(\x,Y) \textrm{ and } \\
&E\left\{\partial^2E(Y\mid\x)/(\partial\x\partial\x\trans)\right\} = \bSig^{-1}\bLam_y\bSig^{-1},
\eeqrs where $\bLam_y \defby E\Big[\left\{Y-E(Y)\right\}(\x-\u)(\x-\u)\trans\Big]$. 
Define $r \defby Y - E(Y) - (\x-\u)\trans\bb$, which is the residual obtained by regressing $Y$ on $\x$ linearly. The Hessians of $E(Y\mid\x)$ and $E(r\mid\x)$ are  equal. Accordingly, we have 
\[E\left\{\partial^2E(Y\mid\x)/(\partial\x\partial\x\trans)\right\} = E\left\{\partial^2E(r \mid\x)/(\partial\x\partial\x\trans)\right\}.\] 
By Stein's Lemma, we can obtain that  
\[E\left\{\partial^2E(r\mid\x)/(\partial\x)(\partial\x\trans)\right\} = \bSig^{-1}\bLam_r\bSig^{-1},\] 
where $\bLam_r \defby E\left\{r (\x-\u)(\x-\u)\trans\right\}$. This indicates that, if $\x$ is normal, we have  explicit forms for  $\bb$ and $\bOme$. Specifically,
\beqrs
&\bb = \bSig^{-1}\cov(\x,Y),\textrm{ and }   
\bOme=  \bSig^{-1}\bLam\bSig^{-1}\big/2, \\
& \textrm{ where } \bLam \textrm{ stands for } \textrm{either }\bLam_y \textrm{ or } \bLam_r.
\eeqrs

We remark here that the normality assumption is widely used in the literature of interaction estimation. See, for example,  \cite{hao2014interaction},  \cite{simon2015permutation}, \cite{bien2015convex}  and \cite{hao2016model}.   In the present context we show that the normality assumption can be relaxed. Let $\tr(\A)$ be the trace operator of matrix $\A = (\A_{k,l})_{p\times p}$. In particular, $\tr(\A) = \sum\limits_{k=1}^p\A_{k,k}$.

\begin{prop} \label{prop1}
	Suppose that $\x$ is drawn from the  factor model $\x = \bGam_0\z + \u$, where 	$\bGam_0$ satisfies $\bGam_0\bGam_0\trans=\bSig>0$ and $\z \defby (Z_1,\ldots,Z_q)\trans$ where $Z_1,\cdots,Z_q$ are independent and identically distributed (i.i.d.) with
	${\rm E}(Z_k) = 0$, ${\rm E}(Z_k^2)=1$, ${\rm E}(Z_k^3)=0$, ${\rm E}(Z_k^4)=\Delta$. We further assume either  
	(C1): $\Delta=3$ or (C2):  $\diag(\bGam_0 \trans \bOme \bGam_0)=\textbf{0}$.
	Then the parameters $\alpha$, $\bb$ and $\bOme$ in model (\ref{model}) have the following explicit forms:
	\beqr\label{explicit}
	&\\
	&\alpha = E (Y)-\tr(\bSig^{-1}\bLam)\big/2, \ \bb =\bSig^{-1}\cov(\x,Y) \textrm{ and }  \bOme=  \bSig^{-1} \bLam \bSig^{-1}\big/2. \nonumber
	\eeqr 
\end{prop}
The factor model was widely assumed in  random matrix theory  \citep{bai1996effect} and high dimensional inference \citep{chen2010tests} where higher order moment assumptions   of $\x$ are quite often required. The moment conditions on $\z$  play an important role to derive  an explicit form  for $\bOme$. Condition (C1) is satisfied if $\x$ is normal. When $\bGam_0 = \I_{p\times p}$, condition (C2) implicitly requires the absence of quadratic terms of the form $X_k^2$ in model (\ref{model}), i.e.,
\beqrs
\E(Y\mid\x) = \alpha + \x \trans \bb + \sum_{i \neq j} \bOme_{i,j} X_i X_j,
\eeqrs
where $X_1,\cdots,X_p$ are i.i.d covariates. 

We provide two explicit forms for estimating $\bOme$, one is based on the response $Y$ and the other is based on the residual  $r$.  The difference between $\bLam_y$ and $\bLam_r$ is that   we remove the main effects in  $\bLam_r$, or equivalently, the linear trend  in model (\ref{model}), before we estimate the interactions $\bOme$. It is thus natural to expect that the residual-based $\bLam_r$ is superior to the response-based $\bLam_y$ in that the sample estimate of $\bLam_r$ has smaller variabilities than that of $\bLam_y$ \citep{cheng2017relative}. In effect, we can replace $\bb$ with an arbitrary  $\wt\bb\in\mR^p$, which yields that $\wt r  \defby Y - E(Y) - (\x-\u)\trans\wt\bb$. Similarly, we can define 
$\bLam_{\wt r} \defby E\left\{\wt r (\x-\u)(\x-\u)\trans\right\}$. Under the normality assumption, $\x$ is symmetric about $\u$ and hence  $\bLam_r = \bLam_{\wt r}$. This ensures that, to estimate  $\bOme$ accurately, our proposal does not hinge  on the sparsity of  main effects because we do not require $\bb$ to be estimated consistently. Even if the main effects are not sufficiently sparse or are not estimated very accurately, we can either directly  use the response-based method $\bSig^{-1}\bLam_y\bSig^{-1}$, or  the  lousy residual-based method $\bSig^{-1}\bLam_{\wt r}\bSig^{-1}$  which utilizes   a lousy residual $\wt r  = Y - E(Y) - (\x-\u)\trans\wt\bb$ and $\wt \bb$ can be a lousy estimate of $\bb$. In effect $\bLam_y$ equals $\bLam_{\wt r}$ by setting $\wt \bb = {\bf0}_{p\times 1}$ in $\wt r$. This makes our proposal  quite different from existing procedures which assume  the heredity conditions and hence require  to estimate the main effects accurately in order to recover the interactions.  By contrast, our proposal does not require to estimate the main effects precisely. We will illustrate this phenomenon  through simulation studies in Section 3.

\subsection{Interaction Estimation}
We show that both  $\bb$ and $\bOme$ have explicit forms under moment conditions in Section 2.1. In particular,   $\bb = \bSig^{-1}\cov(\x,Y)$ and $\bOme = \bSig^{-1}\bLam\bSig^{-1}/2$ for $\bLam$ being $\bLam_y$ or $\bLam_r$. In this subsection, we discuss how to estimate $\bSig^{-1}\cov(\x,Y)$ and $\bSig^{-1}\bLam\bSig^{-1}$ at the sample level.  Estimating $\bSig^{-1}\cov(\x,Y)$  is indeed straightforward by noting that it is a solution to the minimization problem  $$\underset{\tiny \b}{\argmin}~\E\{Y- \E(Y) - (\x-\u)\trans\b\}^2.$$ Therefore,  we can simply estimate $\bSig^{-1}\cov(\x,Y)$ with the penalized least squares  by regressing $\{Y- \E(Y)\}$ on the ultrahigh dimensional covariates $(\x - \u)$ linearly. We do not give  many details about how to estimate $\bSig^{-1}\cov(\x,Y)$  because the penalized least squares estimation has already been well documented \citep{tibshirani1996regression,fan2001variable}. Throughout our numerical studies we use the LASSO \citep{tibshirani1996regression} to estimate $\bb$. The resulting solution is denoted  by $\wh\bb$.

In what follows we concentrate on how to estimate $\bSig^{-1}\bLam\bSig^{-1}/2$, where $\bLam$ can be $\bLam_y$ or $\bLam_r$. For an arbitrary matrix $\B = (\B_{k,l})_{p\times p}$, we have
\beqrs
\bOme&=&\underset{\tiny\B}{\argmin} \Big[\tr(\B-\bSig^{-1}\bLam\bSig^{-1}\big/2)\trans (\B-\bSig^{-1}\bLam\bSig^{-1}\big/2) \Big]\\ &=&\underset{\tiny\B}{\argmin} \Big[\tr(\B-\bSig^{-1}\bLam\bSig^{-1}\big/2)\trans\bSig (\B-\bSig^{-1}\bLam\bSig^{-1}\big/2)\bSig \Big], 
\eeqrs
and 
\beqrs
&\tr(\B-\bSig^{-1}\bLam\bSig^{-1}\big/2)\trans\bSig (\B-\bSig^{-1}\bLam\bSig^{-1}\big/2)\bSig \\
&= \tr(\B \trans \bSig \B \bSig) - \tr(\B\bLam) + \tr(\bSig^{-2}\bLam^2)/4.
\eeqrs
Ignoring the constant, the term $ \tr(\B \trans \bSig \B \bSig) - \tr(\B\bLam)$ quantifies the distance between $\B$  and $\bSig^{-1}\bLam\bSig^{-1}\big/2$. Therefore, to seek a $p\times p$ matrix $\B$ which can  approximate $\bSig^{-1}\bLam\bSig^{-1}\big/2$ very well, it suffices to consider  the following minimization problem
\beqrs
\underset{\tiny\B}{\argmin}\Big[\tr(\B \trans \bSig \B \bSig) - \tr(\B\bLam) \Big],
\eeqrs
as long as we have faithful estimates of $\bSig$ and $\bLam$. The above loss function of matrix form is convex which guarantees that local minimum must be a global minimum.

To construct faithful estimates for  $\bSig$ and $\bLam$,  suppose $\{(\x_i,Y_i),i=1,\ldots,n\}$ is a  random sample of $(\x,Y)$. Denote
\beqrs
&&\overline{\x} \defby n^{-1}\sum_{i=1}^n \x_i,~\overline{Y} \defby n^{-1}\sum_{i=1}^n Y_i,~\hSig \defby
n^{-1}\sum_{i=1}^n \left(\x_i - \overline{\x}\right) \left(\x_i - \overline{\x}\right)\trans, 
\\
&&\wh\bLam = \wh\bLam_y \textrm{ or } \wh\bLam_r, \\
&&\hLam_y \defby
n^{-1}\sum_{i=1}^n (Y_i-\overline{Y} )\left(\x_i - \overline{\x}\right) \left(\x_i - \overline{\x}\right)\trans, \textrm{ and } \\
&&\hLam_r \defby
n^{-1}\sum_{i=1}^n \wh r_i\left(\x_i - \overline{\x}\right) \left(\x_i - \overline{\x}\right)\trans,\eeqrs
where $\wh r_i \defby Y_i - \overline{Y} - (\x_i-\overline{\x})\trans\wh\bb$.  We  propose the following penalized interaction estimation (PIE) to estimate $\bOme$, for $\wh \bLam$ being $\wh \bLam_y$ or $\wh \bLam_r$: 
\beqr\label{pie}
\textrm{\ PIE:\ }~~\hOme=\argmin_{{\tiny\B} \in \mR^{p \times p}} \tr (\B \trans \hSig \B \hSig)-\tr (\B \hLam)+\lambda_{n} \|\B\|_1,
\eeqr
where  $\lambda_{n}$ is a tuning parameter  and  $\|\B\|_{1} \defby \sum\limits_{k=1}^p\sum\limits_{l=1}^p|\B_{k,l}|$.  To ease subsequent illustration, we further define the following two notations:
\beqr\label{piey}
\textrm{\ PIE$_y$:\ }~~\hOme_y&=&\argmin_{{\tiny\B} \in \mR^{p \times p}} \tr(\B \trans \hSig \B \hSig)-\tr (\B \hLam_y)+\lambda_{1n} \|\B\|_1, \textrm{ and }\\\label{pier}
\textrm{\ PIE$_r$:\ }~~\hOme_r&=&\argmin_{{\tiny\B} \in \mR^{p \times p}} \tr (\B \trans \hSig \B \hSig)-\tr (\B \hLam_r)+\lambda_{2n} \|\B\|_1.
\eeqr

\subsection{Implementation} \label{algo}
In this section we discuss how to solve (\ref{pie}) which includes (\ref{piey}) and (\ref{pier}) as special cases. 
Making use of the matrix structure of \eqref{pie}, we next develop an efficient algorithm using the Alternating Direction Method of Multipliers  \citep[ADMM]{boyd2011distributed}. We rewrite the optimization problem in \eqref{pie} as 
\beqr
\min_{\B\in \mR^{p\times p}}  \tr (\B \trans \hSig \B \hSig)-\tr (\B \hLam)+\lambda_n \| \bPsi\|_1, \textrm{ such that } ~\bPsi=\B, \label{admm0}
\eeqr
which motivates us to form the augmented Lagrangian as
\beqr\label{admm}
L(\B,\bPsi,\L)&=&\tr (\B \trans \hSig \B \hSig)-\tr (\B \hLam)+\lambda_n \| \bPsi\|_1\\
&&+
\tr \left\{\L (\B-\bPsi)\right\}+ (\rho/2)\|\B-\bPsi\|_F^2, \nonumber
\eeqr
where $\rho$ is a step size parameter in the ADMM algorithm, and $\|\A\|_F \defby \{\tr(\A\trans\A)\}^{1/2}$ stands for the Frobenius norm of $\A$. Given the current estimate $(\B^k, \bPsi^k, \L^k)$, the augmented Lagrangian \eqref{admm} can   be solved by successively updating $(\B, \bPsi, \L)$ by:
\beqr 
&\textrm{The }\B\textrm{ step: \ \ \ }	\B^{k+1}=&\argmin_{\B \in \mR^{p \times p}}L(\B,\bPsi^k,\L^k),\label{AlgoB} \\&\textrm{The }
\bPsi\textrm{ step: \ \ \ }	\bPsi^{k+1}=&\argmin_{{\small\bPsi} \in \mR^{p \times p}}L(\B^{k+1},\bPsi,\L^k),\label{Algo:Psi}\\  
&\textrm{The }
\L\textrm{ step: \ \ \ }	\L^{k+1}=&\L^k+{\rho}(\B^{k+1}-\bPsi^{k+1}).\label{Algo:L}
\eeqr 
Define the elementwise soft thresholding  operator $\mbox{soft}(\A,\lambda) \defby \{\max(\A_{k,l}-\lambda,0)\}_{p\times p}$. For the $\bPsi$ step, given $\B^{k+1}$,  $\L^k$, $\rho$ and $\lambda_n$, the solution is then given by 
\beqrs
\bPsi^{k+1}\defby\mbox{soft}(\B^{k+1}+ \rho^{-1} \L^{k}, \lambda_n/\rho).\label{eq:UpdatePsi}
\eeqrs

The $\B$ step amounts to solving the equation
\beqr \label{Algo:B}
2 \hSig \B^{k+1} \hSig+\rho \B^{k+1}=\bLam^k,
\eeqr
where $\bLam^k\defby \hLam-\L^k+\rho\bPsi^k$. We make the singular value decomposition to obtain $\hSig=\U \D_0 \U \trans$, where $\U \in \mR^{p \times m}$, $m=\min(n,p)$ and $\D_0\defby \mbox{diag}(d_1,\cdots,d_m)$ is a diagonal matrix.  Define $\D\defby (\D_{k,l})_{p \times p}$, where $\D_{k,l}\defby 2 d_kd_l/(2d_k d_l+\rho)$. Given $\bPsi^k$, $\L^k$ and $ \rho$, the solution to (\ref{Algo:B}) is given by
\beqrs
\B^{k+1}= \rho^{-1}\bLam^k-\rho^{-1} \U \{\D \circ (\U\trans \bLam^k\U)\}\U\trans.  \label{eq:UpdateB}
\eeqrs
where $\circ$ denotes the Hadamard product. 

Details of the algorithm is summarized in Algorithm 1. This algorithm yields a symmetric estimate of $\bOme$, which is denoted by $\hOme$.   The computational complexity of each iteration is no more than O$\{\min(n,p)  p^2\}$ and  the memory requirement is no more than O$(p^2)$ since we only need to store a few $p \times p$ or $p \times \min(n,p)$ matrices in computer memory.
The algorithm explores the advantages of matrix multiplications and is efficient  in  memory storage and computation cost and hence is  appealing for high dimensional quadratic regression. 
\begin{algorithm}
	\caption{Alternating Direction Method of Multipliers (ADMM) for solving \eqref{pie}}
	\begin{algorithmic}[1]
		
		\item[Initialization:]
		\State Input $\{(\x_i,Y_i),i=1,\cdots,n\}$, the tuning parameter $\lambda_n$ and the step size $\rho$;
		\State Calculate $\hLam$ and the singular value decomposition of the centered design matrix $(\x_1-\overline{\x},\cdots,\x_n-\overline{\x})_{p \times n}$ to get $\hSig=\U \D_0 \U \trans$ where $\U \in \mR^{p \times m},~\D_0=\mbox{diag}\{d_1,\dots,d_m\}$ and $m=\min(n,p)$; 
		\State  Define $\D\defby (\D_{k,l})_{m \times m}$ where $\D_{k,l}=2d_kd_l/(2d_k d_l+\rho)$;
		\State Start from $k=0$, $\L^0=\textbf{0}_{p \times p},\B^0=\textbf{0}_{p \times p}$.
		
		\item[Iteration:] 	
		\State Define $\bLam^{k}\defby \hLam-\L^{k}+\rho \B^{k}$. Update  $\B^{k+1}= \rho^{-1}\bLam^k-\rho^{-1} \U \{\D \circ (\U\trans \bLam^k\U)\}\U\trans$;
		\State Update $\bPsi^{k+1}\defby \mbox{soft}(\B^{k+1}+ \rho^{-1} \L^{k}, \lambda_n/\rho)$;
		\State Update $\L^{k+1}\defby\L^{k}+\rho(\B^{k+1}-\bPsi^{k+1})$;
		\State Update $k=k+1$;
		\State Repeat step 5 through step 8 until convergence.
		\item[Output:] $\hOme=\B^{k+1}$.
	\end{algorithmic}
\end{algorithm}

Furthermore, as a first-order method for convex problems, 
convergence analysis of the ADMM algorithm under various conditions has been well documented in the recent optimization literature. See, for example, \cite{nishihara2015general}, \cite{hong2017linear} and \cite{chen2017note}. The following lemma states that our proposed ADMM algorithm  converges linearly to zero.
\begin{lemm}\label{linearADMM}
	Given $\hSig$ and $\hLam$. Suppose that the ADMM algorithm \eqref{AlgoB}-\eqref{Algo:L} generates a   sequence of solutions $\{(\B^k,\bPsi^k, \L^k), k=1,\ldots\}$. Then  $\{(\B^k,\bPsi^k),k=1,\ldots\}$ converges linearly to the minimizer of \eqref{admm0}, and $\|\B^k-\bPsi^k\|_F$ converges linearly to zero, as $k\to\infty$. 
\end{lemm}

It remains to choose an appropriate tuning parameter for PIE$_y$ or PIE$_r$. Motivated by LARS–OLS hybrid \citep{efron2004least}, we use PIE to find the model but not to estimate the coefficients. For a given $\lambda_{n}$, we fit a least squares model on the support of $\hOme$ estimated by PIE$_y$ or PIE$_r$ and get the residual sum of squares.  We then choose $\lambda_{n}$ by the Bayesian information criterion (BIC).  Our limited experience indicates that this procedure is very fast and effective.  

\subsection{Asymptotic Properties}
Suppose $\bOme = (\bOme_{k,l})_{p\times p}$ is a sparse matrix. For notational clarity, we denote the  support of $\bOme$ by  $\supp\defby\{(k,l): \bOme_{k,l}\neq 0\}$,  the complement of $\supp$ by $\supp^c$, and the cardinality of $\calS$ by $s_p\defby \|\bOme\|_0$.  Similarly, we denote by   $\wh\supp_y$ and $\wh\supp_r$ the respective support of $\hOme_y$ and   $\hOme_r$,   and 
$\wh\supp_y^c$ and $\wh\supp_r^c$ the respective complement of $\wh\supp_y$ and $\wh\supp_r$.   We define
$\|\A\|_F \defby \{\tr(\A\trans\A)\}^{1/2}$,  $\|\A\| \defby \sup\limits_{\a\trans\a=1} (\a\trans\A\trans\A\a)^{1/2} = \lambda_{\max}^{1/2}(\A\trans\A)$,
$\|\A\|_{\infty} \defby \max\limits_{1\le k,l\le p}|\A_{k,l}|$ and $\|\A\|_{L} \defby\max\limits_{1\le k \le q}\sum\limits_{l=1}^q|\A_{k,l}|$, for $\A = (\A_{k,l})_{p\times p}$.
We further define $\bGam_0\defby \bSig\otimes \bSig$, $M\defby\|\bGam^{-1}_{\supp,\supp}\|_L$ and $\kappa\defby 1-\|\bGam_{\supp^c,\supp}\bGam_{\supp,\supp}^{-1}\|_L$.
Denote $c_0,C_0, c_1, C_1, \ldots,$   a sequence of generic constants which may take different values at various places.  We assume the following regularity conditions to study the asymptotic properties of $\wh\bOme_y$ and $\wh\bOme_r$.

\begin{enumerate}
	\item[(A1):] Assume $c_0^{-1}\le \lambda_{\min}(\bSig)\le \lambda_{\max}(\bSig)\le c_0$, where $\lambda_{\min}(\bSig)$ and $\lambda_{\max}(\bSig)$ are the respective smallest and largest eigenvalues of $\bSig$.
	\item[(A2):] Assume $X_k$s are sub-Gaussian, i.e.,  $\E\{ \exp (c_0 |\e\trans\x |^2) \}\le C_0 < \infty$ for any unit-length vector $\e$.
	
	\item[(A3)] Assume $\E \{\exp (c_1 |Y|^\alpha)\}\le C_1 < \infty$ for some $0<\alpha \leq 2$.
	\item[(A4)] Assume the irrepresentability condition holds, i.e., $\kappa > 0$.
	\item [(A5)] Assume $\x$ is symmetric about $\u$.
\end{enumerate}
Conditions (A1) and (A2) are widely assumed in high dimensional data analysis. Condition (A3) is assumed to control the tail behavior of $Y$ through concentration inequalities.  The  irrepresentability condition  (A4) is nearly necessary for the consistence of $\ell_1$-penalization \citep{zhao2006model, zou2006adaptive}. This condition was first used by \cite{ravikumar2011high}. See also  \cite{zhang2014sparse} and \cite{liu2015fast}. We assume condition (A5) to ensure  the consistency of residual-based approaches.

\begin{theo}\label{thm1}
	Let  $\lambda_{1n}\defby c_1   \{n^{-\alpha/(\alpha+1)}\log(p)\}^{1/2}$ for sufficiently large $c_1$ and assume that $s_p \{n^{-1}\log (p)\}^{1/2}\rightarrow 0$. Under the conditions (A1)-(A4), we have
	\begin{itemize}
		\item[(i)] ${\rm pr}\big(\wh\supp_y^c=\supp^c\big)= 1-O(p^{-1})$.
		\item[(ii)] If we further assume 
		$\min\limits_{(k,l) \in \supp}|\bOme_{k,l}|>c_2 M \lambda_{1n}$
		for sufficiently large $c_3$, then ${\rm pr} \big(\wh\supp_y= \supp \big)=1-O(p^{-1})$.
		\item[(iii)] ${\rm pr}\big(	\|\hOme_y-\bOme\|_\infty\leq c_3  \lambda_{1n}M\big) = 1-O(p^{-1})$, for sufficiently large $c_3$.
		\item[(iv)] ${\rm pr}\big(	\|\hOme_y-\bOme\|_F\leq c_4  s_p^{1/2}\lambda_{1n}M\big)= 1-O(p^{-1})$, for sufficiently large $c_4$.
	\end{itemize}
\end{theo}

Theorem \ref{thm1} shows that,  as long as the signal strength of the interactions is not too small, our proposal can identify the support correctly with a very high probability. In other words, $\hOme_y$ is asymptotically selection consistent. Theorem \ref{thm1} also shows that $\hOme_y$ is a  consistent estimate of $\bOme$ under both the infinity norm and the Frobenius norm.

\begin{theo}\label{thm2}
	Let $\lambda_{2n}\defby c_5 \{n^{-\alpha/(\alpha+1)}\log(p)\}^{1/2}+c_5 \|\wh \bb-\bb\|_1 \{\log(p)/n\}^{1/2}$ for sufficiently large $c_5$ and assume that $s_p \{n^{-1}\log (p)\}^{1/2}\rightarrow 0$. Under the conditions (A1)-(A5), we have
	\begin{itemize}
		\item[(i)] ${\rm pr}\big(\wh\supp_r^c=\supp^c\big)= 1-O(p^{-1})$.
		\item[(ii)] If we further assume 
		$\min\limits_{(k,l) \in \supp}|\bOme_{k,l}|>c_6 M \lambda_{2n}$
		for sufficiently large $c_6$, then ${\rm pr} \big(\wh\supp_r= \supp \big)=1-O(p^{-1})$.
		\item[(iii)] ${\rm pr}\big(	\|\hOme_r-\bOme\|_\infty\leq c_7  \lambda_{2n}M\big)= 1-O(p^{-1})$, for sufficiently large $c_7$.
		\item[(iv)] ${\rm pr}\big(	\|\hOme_r-\bOme\|_F\leq c_8  s_p^{1/2}\lambda_{2n}M\big)= 1-O(p^{-1})$, for sufficiently large $c_8$.
	\end{itemize}
\end{theo}
Theorem \ref{thm2} shows that $\hOme_r$, as well as $\hOme_y$,   possesses both the  selection  and estimation consistency asymptotically. 
Moreover, the convergence rate of $\hOme_r$ depends on  $\wh \bb$. If $\|\wh \bb-\bb\|_1=o\{n^{1/(2 \alpha+2)}\}$, the convergence rate term involving $\wh \bb$ will be absorbed in the first term of Theorem \ref{thm2}.  In other words, unless the estimation error of $\wh \bb$ diverges faster than $n^{1/(2 \alpha+2)}$,  $\hOme_r$ and $\hOme_y$ would share the same convergence rate. 

\subsection{Connections to All-Pairs-LASSO}
For quadratic regression, a nature way is to fit LASSO model on all pairs of interactions, 
\begin{align*}
\argmin_{\alpha,\bb,\B} (2n)^{-1}(Y_i-\alpha-\x_i \trans \bb -\x_i \trans \B \x_i)^2+\lambda_n (\|\bb\|_1+ \|\B\|_1).
\end{align*}
Following \cite{bien2013lasso}, we refer to this approach as the all-pairs-LASSO.  For brevity, we assume $E(\x)=\textbf{0}$ and ignore the main effects.  Write $\z_i\defby \x_i \otimes \x_i$ and $\overline{\z}\defby n^{-1}\sum_{i=1}^n \z_i$. The all-pairs-LASSO is equivalent to
\begin{eqnarray}\label{compare1}
\argmin _{\B}&&(2n)^{-1}\sum_{i=1}^n\{(Y_i-\overline{Y})-(\z_i-\overline{\z} ) \trans \vec(\B)\}^2+\lambda_n \|\B\|_1   \\
=\argmin _{\B}&&(2n)^{-1}  \vec(\B) \trans \sum_{i=1}^n(\z_i-\overline{\z} ) (\z_i-\overline{\z} ) \trans \vec(\B)\nonumber\\
&&-n^{-1}\sum_{i=1}^n(Y_i-\overline{Y})(\z_i-\overline{\z} ) \trans \vec(\B)+\lambda_n \| \vec(\B)\|_1. \nonumber
\end{eqnarray}
where $\otimes$ denotes the Kronecker product and $\vec(\cdot)$ stands for the vectorization of a matrix. Recall that  our proposed method  can be re-expressed as 
\beqr\label{compare2}
\vec{(\hOme_y)}&=&\argmin_{\B} 2^{-1}\vec{(\B)}  \trans (2 \hSig \otimes \hSig) \vec{(\B)} \\
&&-\vec{(\hLam_y)} \trans \vec{(\B)}+\lambda_n \|\vec{(\B)}\|_1. \nonumber
\eeqr
It is straightforward to show that $n^{-1}\sum_{i=1}^n(Y_i-\overline{Y})(\z_i-\overline{\z} )=\vec{(\hLam_y)}$. Plug  this into \eqref{compare1} and compare with   \eqref{compare2}. We can see that the only difference between all-pairs-LASSO and our proposed PIEy is the first term. The all-pairs-LASSO  directly uses the sample version $n^{-1}  \sum_{i=1}^n(\z_i-\overline{\z} ) (\z_i-\overline{\z} ) \trans$  to mimic the covariance structure  $\cov(\z)$ while our method propose to use $2 \hSig \otimes \hSig$ since we have $\cov\{\z \trans \vec(\B)\}=2 \tr(\B \trans \bSig \B \bSig)$ under the moment condition of  Proposition \ref{prop1}.

Using $\hSig \otimes \hSig$ gives at least two advantages.  The first   is the computational efficiency. 
The complexity of our proposed method is  $O\{\min(n,p) p^2\}$. When $p$ is larger than $n$,  the computation complexity is linear in both $n$ and the number of parameters which is of order $p^2$. Comparing with all-pairs-LASSO, the memory our proposed method required is much less. In all-pairs-LASSO, we need to store $O(p^2) \times n$ design matrix where our methods only depends on several $p \times p$ matrices. The second advantage is on the theoretical properties. Under mild conditions, we can show that
\begin{align*}
\|\hSig \otimes \hSig-\bSig \otimes \bSig\|_\infty \approx \|\hSig-\bSig\|_\infty= O_p\{n^{-1}\log (p)\}^{1/2}.
\end{align*} 
By  Lemma 2 in Appendix,
\begin{align*}
\|n^{-1}  \sum_{i=1}^n(\z_i-\overline{\z} ) (\z_i-\overline{\z} ) \trans-\cov(\z)\|_\infty=O_p\{n^{-1/2}\log (p)\}^{1/2}.
\end{align*}
It can be  seen that using   $\hSig \otimes \hSig$ gives a better convergence rate. 
We will demonstrate these issues through simulations in the next section.

\section{SIMULATIONS}
In this section we conduct simulations to evaluate the performance of our proposal and 
to compare it with the RAMP method  \citep{hao2016model} and the all-pairs-LASSO which fits a LASSO model on all $p$ main effects and $p(p+1)/2$ interactions. By \cite{hao2016model}, RAMP will outperforms other methods such as iFOR\citep{hao2014interaction} and hierNet \citep{bien2013lasso} under heredity assumptions and hence in our simulations we only include RAMP as a representative. In what follows, we refer to the RAMP method under the strong heredity condition as ``RAMPs" and  the RAMP method under the weak heredity condition as ``RAMPw".  We also include the oracle estimate as a benchmark which assumes the main effects and the support of interactions are  known in advance. The oracle estimate simply fits the least squares estimation on the support of interactions using the truly important main effects and we denote it as ``Oracle". The RAMP method and all-pairs-LASSO are implemented by the R packages ``RAMP" and  ``glmnet" \citep{glmnet}.  The developed R package ``PIE" which implements our proposal is available online. 

To ease illustration, we denote the estimate of $\bOme$ by $\hOme$  obtained with different approaches. We evaluate the accuracy of the estimation through three criteria: the support recovery rate, denoted by ``rate",  the Frobenius loss, denoted by ``loss" and the number of interactions that are estimated as nonzero, denoted  by ``size". To be specific, the criteria are defined as follows,
\beqrs
\mbox{rate}&\defby& \sum_{l \leq k} I(\hOme_{k,l} \neq 0, \bOme_{k,l} \neq 0 )\Big/\sum_{l \leq k} I(\bOme_{k,l} \neq 0)\times100\%,\\
~~\mbox{loss}&\defby& \|\hOme-\bOme\|_2,~\mbox{and~~size}\defby\sum_{l \leq k} I(\hOme_{k,l} \neq 0).
\eeqrs
Here $I(E)$ is an indicator function which equals 1 if the random event $E$ is true and 0 otherwise. The closer the ``rate" is to one, the ``loss" is to zero and the ``size" is to the number of truly important interactions, the better performance a proposal has.

We consider the following four models. 
\begin{eqnarray}\label{m1}
Y&=&X_1+X_6+X_{10}+2X_{1} X_{6}+X^2_{6}+2X_{6} X_{10}+\varepsilon,\\\label{m2}
Y&=&X_6+2X_{1} X_{6}+X^2_{6}+2X_{6} X_{10}+\varepsilon,\\\label{m3}
Y&=&X_1+X_2+2X_{1} X_{6}+X^2_{6}+2X_{6} X_{10}+\varepsilon,\\\label{m4}
Y&=&2X_{1} X_{6}+X^2_{6}+2X_{6} X_{10}+\varepsilon.
\end{eqnarray}
The strong heredity condition holds in model (\ref{m1}) and the weak heredity condition holds in model (\ref{m2}), respectively. Neither the strong nor the weak heredity condition holds in model  (\ref{m3}) or (\ref{m4}). In particular, model (\ref{m4}) is a pure interaction model. We replicate each scenario 100 times to evaluate the performance of different proposals.

\subsection{Estimation Accuracy} \label{sub1}
We draw $\x$ independently from ${\cal N}({\bf0}_{p\times 1},\bSig)$ where $\bSig$ is the power decay covariance matrix $(0.5^{|k-l|})_{p \times p}$ and generate an independent error $\varepsilon$ from ${\cal N}(0,1)$.  We set the sample size $n=200$ and the  dimension $p=100$ or $p=200$.

The simulation results are charted in Tables \ref{tab1}.  We can observe that our proposal has a stable performance across almost all scenarios.  It is not very surprising to see that, the RAMP method with strong heredity condition, denoted RAMPs, completely fails in models (\ref{m2})-(\ref{m4}) where the strong heredity condition is violated; in addition, the RAMP method with weak heredity condition, denoted RAMPw,  fails in models (\ref{m3})-(\ref{m4}) where the weak heredity condition is also violated. The RAMP method has a satisfactory performance when the required heredity condition is satisfied.  In particular, the RAMPs performs quite well in model (\ref{m1}). For models (\ref{m2})-(\ref{m4}), the oracle estimate has the smallest Frobenius loss, followed by our proposals. Comparing with the all-pairs-LASSO, under all the settings, our proposal has a better performance in terms of Frobenious loss and model size. For the pure interaction model (\ref{m4}) where no main effects are present, fitting linear regression to obtain  residuals very likely introduces some redundant bias. It is thus not surprising to see that  our proposed response-based procedure (PIEy) slightly outperforms our  residual-based procedure (PIEr).

\begin{table}[!htbp] \centering 
	\caption{The averages (and standard deviations) of  the support recovery rate (``rate"), the Frobenius loss (``loss") and the model size (``size") for  models (\ref{m1})-(\ref{m4}).}
	\label{tab1} 
	\resizebox{1\textwidth}{!}{%
		\begin{tabular}{cccccccc} 
			\\[-1.8ex] 
			\hline 
			\hline 
			$p$ &   & PIEy & PIEr & RAMPs & RAMPw & all-pairs-LASSO & Oracle \\ 
			\hline 
			&\multicolumn{7}{c}{model (\ref{m1}) where the strong heredity condition is satisfied}	\\
			100 & rate &  99.33(4.69) &  99.67(3.33) &  85.00(35.89) &  97.33(12.25) & 100.00(0.00) & 100.00(0.00) \\ 
			& loss &   0.33(0.21) &   0.22(0.14) &   0.42(0.81) &   0.17(0.34) &   0.37(0.08) &   0.09(0.04) \\ 
			& size &   4.31(2.21) &   3.55(0.87) &   3.03(1.67) &   3.38(1.56) &   9.36(5.50) &   3.00(0.00) \\ 
			200 & rate &  98.33(7.30) &  99.33(4.69) &  88.00(31.61) &  99.00(7.42) & 100.00(0.00) & 100.00(0.00) \\ 
			& loss &   0.43(0.30) &   0.29(0.22) &   0.37(0.72) &   0.13(0.24) &   0.43(0.10) &   0.09(0.04) \\ 
			& size &   5.57(3.66) &   4.79(4.34) &   2.83(1.15) &   3.62(2.70) &  10.21(7.72) &   3.00(0.00) \\ 
			\hline 
			&\multicolumn{7}{c}{model (\ref{m2}) where the weak heredity condition is satisfied}	\\
			100 & rate & 100.00(0.00) & 100.00(0.00) &  35.33(24.07) &  86.33(34.20) & 100.00(0.00) & 100.00(0.00) \\ 
			& loss &   0.18(0.08) &   0.17(0.09) &   1.93(0.43) &   0.41(0.81) &   0.40(0.10) &   0.08(0.04) \\ 
			& size &   3.64(1.37) &   3.54(1.27) &   1.85(2.43) &   3.91(3.96) &   5.31(2.55) &   3.00(0.00) \\ 
			200 & rate &  98.33(7.30) &  99.00(5.71) &  35.00(20.85) &  86.00(34.87) & 100.00(0.00) & 100.00(0.00) \\ 
			& loss &   0.24(0.24) &   0.22(0.19) &   1.95(0.41) &   0.42(0.82) &   0.43(0.11) &   0.09(0.04) \\ 
			& size &   4.17(3.18) &   4.45(4.44) &   1.48(1.49) &   3.68(3.05) &   5.93(2.98) &   3.00(0.00) \\ 
			\hline 
			&\multicolumn{7}{c}{model (\ref{m3}) where the heredity conditions is violated}	\\
			100 & rate &  99.00(5.71) & 100.00(0.00) &  21.33(34.98) &  46.00(26.29) & 100.00(0.00) & 100.00(0.00) \\ 
			& loss &   0.30(0.23) &   0.17(0.10) &   1.93(0.62) &   1.51(0.69) &   0.39(0.10) &   0.09(0.04) \\ 
			& size &   4.65(3.31) &   3.64(1.55) &   1.12(1.85) &   3.83(4.48) &   7.17(5.40) &   3.00(0.00) \\ 
			200 & rate &  98.67(6.56) &  99.33(4.69) &  16.00(29.01) &  41.33(23.27) & 100.00(0.00) & 100.00(0.00) \\ 
			& loss &   0.36(0.24) &   0.21(0.17) &   2.06(0.39) &   1.60(0.58) &   0.42(0.09) &   0.09(0.04) \\ 
			& size &   4.97(2.63) &   3.88(2.05) &   0.88(1.47) &   2.84(3.68) &   7.62(5.68) &   3.00(0.00) \\ 
			\hline
			&\multicolumn{7}{c}{model (\ref{m4}) is a pure interaction model where the heredity conditions are violated}	\\		
			100 & rate & 100.00(0.00) & 100.00(0.00) &  13.33(24.16) &  28.00(42.30) & 100.00(0.00) & 100.00(0.00) \\ 
			& loss &   0.11(0.05) &   0.14(0.08) &   2.13(0.36) &   1.73(0.94) &   0.41(0.10) &   0.10(0.04) \\ 
			& size &   3.48(1.03) &   3.54(1.10) &   1.06(2.06) &   3.49(5.27) &   5.05(3.85) &   3.00(0.00) \\ 
			200 & rate &  99.33(4.69) &  99.33(4.69) &   6.67(18.35) &  15.67(34.96) & 100.00(0.00) & 100.00(0.00) \\ 
			& loss &   0.12(0.14) &   0.14(0.15) &   2.17(0.27) &   1.97(0.78) &   0.45(0.09) &   0.09(0.04) \\ 
			& size &   3.68(2.97) &   3.68(2.88) &   0.29(0.81) &   2.72(5.14) &   4.61(2.55) &   3.00(0.00) \\ 
			\hline 
	\end{tabular} }
\end{table}

\subsection{Ultrahigh Dimensional Covariates} 
Our algorithm is very efficient with cheap computation complexity and computer memory. In this part, we demonstrate the performance of our proposal under ultrahigh dimension settings. Apart from the three criteria considered in the previous subsection, we also compare the computation time among all the methods to illustrate the computation efficiency of our method. The parameter settings are the same as those in Subsection \ref{sub1} except that the data dimension $p$ is now set to be 500, 1000 or 2000, and the sample size $n$ is set to be 400 or 800. To save space, we only report the results for model \eqref{m2} where the weak heredity condition holds.

Table \ref{tab2} summaries the simulations results including the ``rate", ``loss",  ``size" and the computation time in seconds (denoted as ``time"). All methods are implemented with a PC with a 3.3 GHz Intel Core i7 CPU and 16GB memory. Overall, the patterns of the estimation accuracy are similar to those in Table \ref{tab1}. For the computation time, it can be seen that our methods are very effective comparing with other methods.
In addition, we can observed that the computation time of our methods increase linearly in $n$ and $p^2$, which is consistent with the computation complexity $O\{\min(n,p) p^2\}$ we claimed in the last section. The computation time of RAMP is not so sensitive to the sample size or data dimension since it used the structure information of heredity conditions.  For the all-pairs-LASSO, 
we test the computation time using LARS \citep{efron2004least} and it turns out to be very slow. We instead 
implemented the all-pairs-LASSO using ``glmnet"  \citep{glmnet}. We remark that ``glmnet"  \citep{glmnet} is the state of art algorithm for LASSO problems and the package was further accelerated by strong rules \citep{tibshirani2012strong}. From Table \ref{tab2} we can see that the computation time also seems to be increasing linearly in $n$ and quadratically in $p$. However, the all-pairs-LASSO uses more computer memory since the number of covariates is of order $O(p^2)$ and will break down when $p=2000$ due to out of memory in R.  In summary, our proposal are more efficient than the all-pairs-LASSO in both computation complexity and computation memory.
\begin{table}[!htbp] \centering 
	\caption{Simulation results for weak heredity model with ultra-high covariates.}
	\label{tab2} 
	\resizebox{1\textwidth}{!}{%
		\begin{tabular}{cccccccc} 
			\\[-1.8ex] 
			\hline 
			\hline 
			$p$ &   & PIEy & PIEr & RAMPs & RAMPw & all-pairs-LASSO & Oracle \\ 
			\hline 
			&\multicolumn{7}{c}{$n=400$}	\\
			500 & rate & 100.00(0.00) & 100.00(0.00) &  38.00(13.42) &  99.00(10.00) & 100.00(0.00) & 100.00(0.00) \\ 
			& loss &   0.14(0.08) &   0.11(0.07) &   1.94(0.27) &   0.09(0.24) &   0.31(0.06) &   0.06(0.02) \\ 
			& size &   3.56(1.29) &   3.15(0.58) &   1.27(0.74) &   3.24(1.91) &   5.11(4.17) &   3.00(0.00) \\ 
			& time &   3.90(0.41) &   3.75(0.33) &  28.71(8.54) &  26.37(5.39) &  32.90(3.43) &   0.02(0.00) \\ 
			1000 & rate & 100.00(0.00) & 100.00(0.00) &  37.00(15.64) &  95.33(20.66) & 100.00(0.00) & 100.00(0.00) \\ 
			& loss &   0.13(0.08) &   0.09(0.05) &   1.93(0.33) &   0.17(0.52) &   0.34(0.06) &   0.06(0.03) \\ 
			& size &   3.60(1.62) &   3.20(0.95) &   1.38(1.20) &   3.76(3.85) &   4.02(2.09) &   3.00(0.00) \\ 
			& time &  12.70(0.49) &  12.56(0.55) &  48.26(8.88) &  50.84(10.54) & 126.66(0.55) &   0.04(0.01) \\ 
			2000 & rate & 100.00(0.00) & 100.00(0.00) &  31.67(11.96) &  88.00(32.66) & -& 100.00(0.00) \\ 
			& loss &   0.15(0.10) &   0.13(0.09) &   2.02(0.14) &   0.34(0.78) & -&   0.06(0.02) \\ 
			& size &   3.83(2.00) &   3.29(0.82) &   1.19(0.92) &   4.67(6.38) & -&   3.00(0.00) \\ 
			& time &  58.46(6.15) &  59.33(6.21) &  34.61(4.62) &  92.41(21.77) & - &   0.19(0.03) \\ 
			\hline		
			&\multicolumn{7}{c}{$n=800$}	\\
			500 & rate & 100.00(0.00) & 100.00(0.00) &  38.67(13.99) & 100.00(0.00) & 100.00(0.00) & 100.00(0.00) \\ 
			& loss &   0.09(0.04) &   0.06(0.04) &   1.92(0.33) &   0.04(0.02) &   0.22(0.04) &   0.04(0.02) \\ 
			& size &   3.20(0.64) &   3.05(0.26) &   1.98(3.08) &   3.04(0.20) &   3.50(1.45) &   3.00(0.00) \\ 
			& time &   4.24(0.55) &   4.13(0.58) & 102.20(15.51) & 132.41(15.95) &  62.85(0.72) &   0.02(0.00) \\ 	
			1000 & rate & 100.00(0.00) & 100.00(0.00) &  38.33(11.96) & 100.00(0.00) & 100.00(0.00) & 100.00(0.00) \\ 
			& loss &   0.09(0.05) &   0.06(0.03) &   1.94(0.21) &   0.04(0.02) &   0.23(0.03) &   0.04(0.01) \\ 
			& size &   3.28(0.98) &   3.06(0.37) &   1.77(2.24) &   3.01(0.10) &   3.41(0.78) &   3.00(0.00) \\ 
			& time &  25.95(2.65) &  25.64(2.42) & 116.46(23.30) & 131.51(25.39) & 261.54(9.76) &   0.06(0.01) \\ 
			2000 & rate & 100.00(0.00) & 100.00(0.00) &  37.00(12.44) & 100.00(0.00) & -& 100.00(0.00) \\ 
			& loss &   0.09(0.06) &   0.06(0.04) &   1.94(0.31) &   0.04(0.02) &   - &   0.04(0.02) \\ 
			& size &   3.52(1.73) &   3.08(0.37) &   1.37(1.32) &   3.01(0.10) &  - &   3.00(0.00) \\ 
			& time &  90.72(6.72) &  96.52(7.18) & 243.33(72.01) & 249.13(46.48) &  - &   0.24(0.02) \\ 
			\hline
			\multicolumn{8}{l}{$-$ out of memory in R}
	\end{tabular} }

\end{table}

\subsection{Estimation of Main Effects} 
In this section we evaluate how estimation of  main effects affects the estimation of interactions.   Both our proposed residual-based penalized interaction estimation and the RAMP method involve estimating the main effects. To fixed the signal-to-noise ratio for all the settings, we simply draw the covariates $\x=(X_1,\ldots,X_p)\trans$ from ${\cal N}({\bf0},\I_{p\times p})$ and consider the following quadratic model
\beqrs
Y =& d^{-1/2}\left(X_1+X_6+X_{10}+X_{k_1}+\cdots+X_{k_{d-3}}\right)\\
&+2X_{1} X_{6}+X^2_{6}+2X_{6} X_{10}+\varepsilon.
\eeqrs
The number of main effects is increased from $d=  3$ to $48$.  We always include $X_1$, $X_6$ and $X_{10}$  to ensure that the strong heredity condition holds true. We also randomly choose $X_{k_1},\ldots, X_{k_{d-3}}$ from $X_{11},\ldots,X_{p}$. Figure  \ref{fig1} reports the support recovery rate  of $\hOme$ and the Frobenius loss of $\|\hOme-\bOme\|_F$.

It can be clearly seen that, as the number of main effects increases from $d=3$ to $48$, both versions of the RAMP method, RAMPs and RAMPw, deteriorate gradually in terms of both criteria, indicating that the RAMP method heavily relies on estimating the main effects accurately. For all-pairs-LASSO, the support recovery rate is good while the Frobenius loss becomes worse when $d$ increases.  By contrast, our proposal is very robust to the number of main effects under both criteria. 
Moreover, when the number of main effects increases, PIEy will be slightly better than PIEr in terms of Frobenius loss. Theses findings confirm our theoretical results in Theorem \ref{thm2} since the estimation $\hb$ will become worse when $d$ increases. 

\begin{figure}[!ht]
	\centerline{
		\begin{tabular}{cc}				
			\psfig{figure=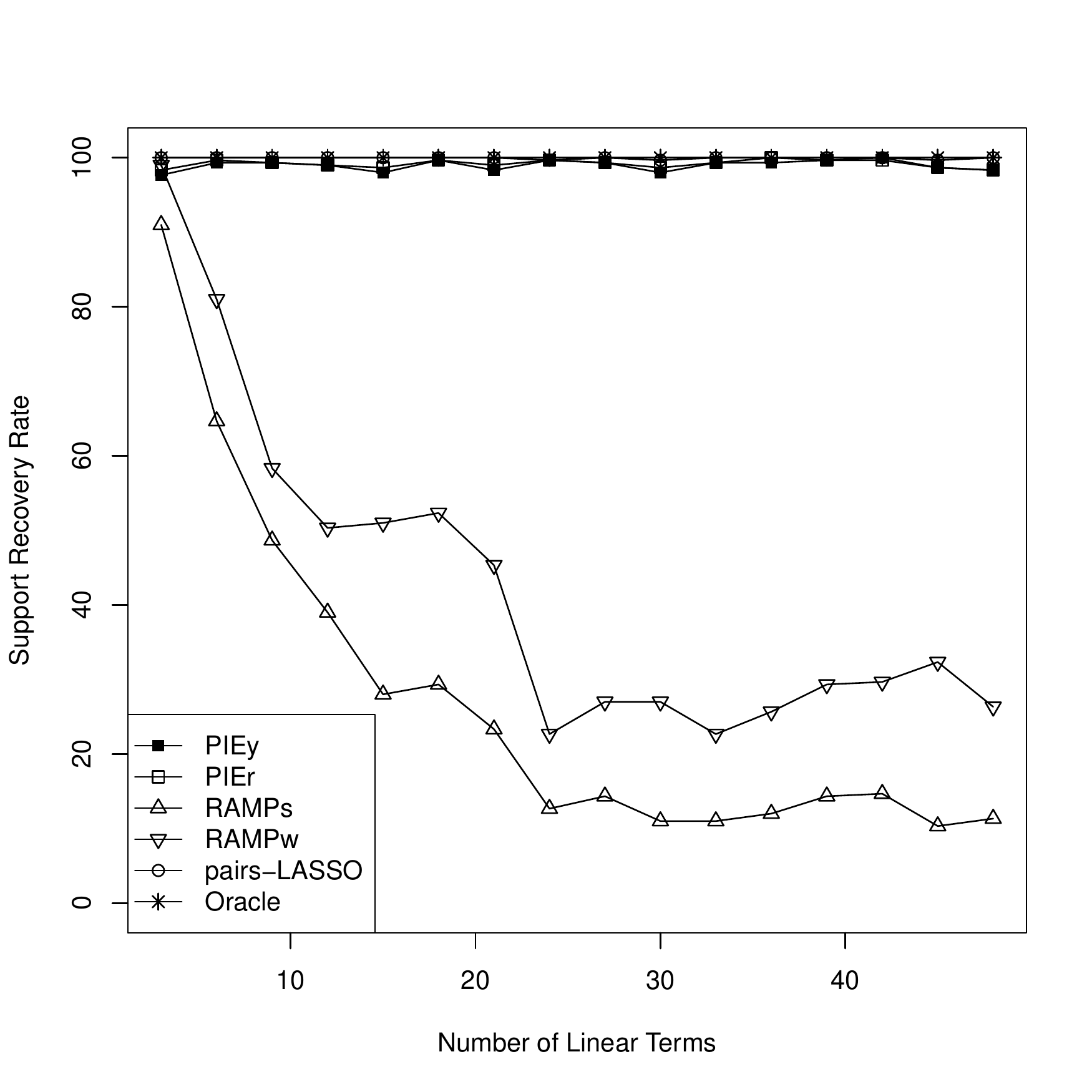,width=2.5 in,height=2.5 in,angle=0} &
			\psfig{figure=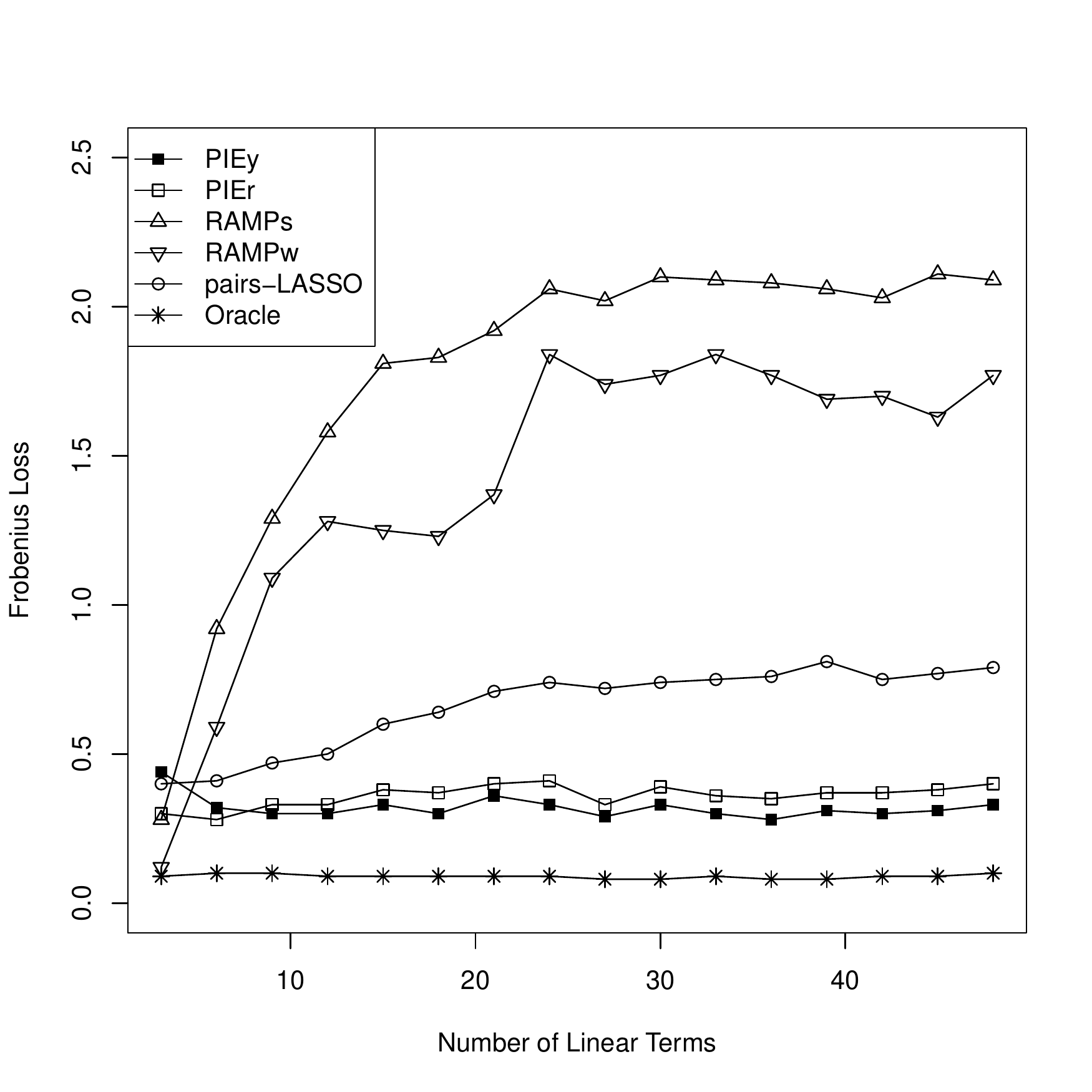,width=2.5 in,height=2.5 in,angle=0} \\
		\end{tabular}
	}
	\caption{\textit{The vertical axis stands for the support recovery rate (left) and Frobenius loss (right) of $\hOme$, and the horizontal axis stands for the number of main effects.}}
	\label{fig1}
\end{figure}

\subsection{Non-Normal Covariates}
In this part, we investigate the performance of our proposal when the covariates are non-normal, and the factor model assumptions are violated. Let $\x=\bSig^{1/2}\z,~\bSig=(0.5^{|i-j})_{100 \times 100}$ and $\z = (Z_1,\cdots,Z_p)\trans$. We draw $Z_k$s independently from (i) uniform distribution on the interval $[-\sqrt 3, \sqrt 3]$ where $\Delta=1.8$, (ii) Student's t-distribution $t(5)\sqrt{3/5}$ where $\Delta=9$ and (iii) Laplace distribution $\mbox{Laplace}(0,1)/\sqrt{2}$ where $\Delta=6$. In all scenarios, the $Z_k$'s are symmetric and have unit variance.  

Table \ref{tab3} 
reports the support recovery rate (``rate") and the number of interactions that are estimated as nonzero (``size") and  the Frobenius loss (``loss") of $\hOme$. From Table \ref{tab3} we can see  
that PIEy and PIEr are still very effective when the covariates are non-normal, and the performance comparing with other methods are similar to those we have observed under normal assumptions, indicating that our proposal is practically robust to the violation of the theoretical assumptions.

\begin{table}[!htbp] 
	\caption{Simulation results for non-normal covariates where $n=400$ and $p=100$.}
	\label{tab3}
	\centering
	\resizebox{1\textwidth}{!}{%
		\begin{tabular}{@{\extracolsep{5pt}} cccccccc} 
			\\[-1.8ex]\hline 
			\hline
			&  & PIEy & PIEr & RAMPs & RAMPw & all-pairs-LASSO & Oracle \\ 
			\hline \\[-1.8ex] 
			\hline 
			&\multicolumn{7}{c}{model (\ref{m1}) where the strong heredity condition is satisfied}	\\
			Unif & rate &  99.33(4.69) &  99.67(3.33) & 100.00(0.00) & 100.00(0.00) & 100.00(0.00) & 100.00(0.00) \\ 
			& size &   3.86(1.73) &   3.19(0.72) &   3.00(0.00) &   3.03(0.17) &   5.96(3.94) &   3.00(0.00) \\ 
			& loss &   0.22(0.18) &   0.13(0.12) &   0.06(0.02) &   0.06(0.03) &   0.26(0.05) &   0.06(0.02) \\ 
			t(5) & rate &  93.33(17.08) &  95.33(14.23) &  93.00(25.64) &  99.00(5.71) & 100.00(0.00) & 100.00(0.00) \\ 
			& size &   6.12(3.35) &   5.99(5.80) &   3.40(2.27) &   3.59(1.93) &   7.61(4.96) &   3.00(0.00) \\ 
			& loss &   0.47(0.55) &   0.33(0.50) &   0.22(0.62) &   0.11(0.27) &   0.25(0.06) &   0.06(0.03) \\ 
			Lap & rate & 100.00(0.00) & 100.00(0.00) &  90.67(28.85) &  98.67(6.56) & 100.00(0.00) & 100.00(0.00) \\ 
			& size &   5.87(4.12) &   4.90(3.61) &   2.93(1.27) &   3.75(3.15) &   7.10(5.27) &   3.00(0.00) \\ 
			& loss &   0.25(0.12) &   0.15(0.06) &   0.27(0.66) &   0.11(0.24) &   0.23(0.06) &   0.06(0.03) \\ 
			\hline 
			&\multicolumn{7}{c}{model (\ref{m2}) where the weak heredity condition is satisfied}	\\
			Unif & rate &  99.67(3.33) &  99.67(3.33) &  46.67(20.65) & 100.00(0.00) & 100.00(0.00) & 100.00(0.00) \\ 
			& size &   3.19(0.61) &   3.14(0.62) &   2.17(2.28) &   3.01(0.10) &   4.30(2.53) &   3.00(0.00) \\ 
			& loss &   0.13(0.11) &   0.11(0.11) &   1.75(0.53) &   0.07(0.04) &   0.28(0.06) &   0.07(0.03) \\ 
			t(5) & rate &  94.33(15.75) &  95.00(14.51) &  51.33(27.80) &  98.00(14.07) & 100.00(0.00) & 100.00(0.00) \\ 
			& size &   6.17(4.74) &   6.14(6.50) &   2.99(2.85) &   3.39(1.98) &   5.20(3.39) &   3.00(0.00) \\ 
			& loss &   0.36(0.55) &   0.31(0.53) &   1.61(0.69) &   0.11(0.36) &   0.25(0.05) &   0.06(0.03) \\ 
			Lap & rate & 100.00(0.00) & 100.00(0.00) &  48.67(28.59) &  94.00(23.87) & 100.00(0.00) & 100.00(0.00) \\ 
			& size &   5.37(3.72) &   5.21(4.17) &   2.54(2.72) &   3.56(3.07) &   4.87(2.41) &   3.00(0.00) \\ 
			& loss &   0.17(0.08) &   0.13(0.08) &   1.63(0.69) &   0.20(0.57) &   0.25(0.06) &   0.05(0.02) \\ 
			\hline
			&\multicolumn{7}{c}{model (\ref{m3}) where the heredity conditions is violated}	\\
			Unif & rate &  99.67(3.33) &  99.67(3.33) &  14.00(29.66) &  46.33(25.47) & 100.00(0.00) & 100.00(0.00) \\ 
			& size &   3.95(1.83) &   3.13(0.44) &   1.19(2.91) &   4.27(5.99) &   5.26(3.89) &   3.00(0.00) \\ 
			& loss &   0.22(0.16) &   0.12(0.12) &   2.05(0.50) &   1.46(0.66) &   0.28(0.06) &   0.06(0.02) \\ 
			t(5) & rate &  94.00(15.98) &  94.33(15.02) &  37.00(41.81) &  68.33(30.84) & 100.00(0.00) & 100.00(0.00) \\ 
			& size &   5.93(3.25) &   5.24(3.01) &   2.99(4.30) &   5.07(5.39) &   6.26(4.24) &   3.00(0.00) \\ 
			& loss &   0.40(0.54) &   0.33(0.55) &   1.65(0.85) &   0.98(0.87) &   0.25(0.07) &   0.06(0.03) \\ 
			Lap & rate &  99.67(3.33) & 100.00(0.00) &  42.00(41.47) &  62.00(31.79) & 100.00(0.00) & 100.00(0.00) \\ 
			& size &   5.80(4.08) &   5.08(4.07) &   2.76(3.16) &   6.02(6.95) &   5.86(3.35) &   3.00(0.00) \\ 
			& loss &   0.21(0.17) &   0.11(0.06) &   1.59(0.85) &   1.12(0.87) &   0.24(0.06) &   0.06(0.03) \\ 
			\hline
			&\multicolumn{7}{c}{model (\ref{m4}) is a pure interaction model where the heredity conditions are violated}	\\		
			Unif & rate &  99.67(3.33) &  99.67(3.33) &   6.67(17.08) &  15.67(32.29) & 100.00(0.00) & 100.00(0.00) \\ 
			& size &   3.08(0.53) &   3.06(0.34) &   0.48(1.42) &   3.28(6.52) &   3.82(1.50) &   3.00(0.00) \\ 
			& loss &   0.08(0.11) &   0.09(0.11) &   2.18(0.17) &   1.98(0.71) &   0.31(0.06) &   0.06(0.03) \\ 
			t(5) & rate &  94.33(15.75) &  94.67(14.77) &  26.00(33.02) &  49.67(46.30) & 100.00(0.00) & 100.00(0.00) \\ 
			& size &   6.00(6.19) &   5.97(6.15) &   1.81(2.92) &   5.66(7.17) &   5.10(4.03) &   3.00(0.00) \\ 
			& loss &   0.29(0.57) &   0.29(0.55) &   1.92(0.61) &   1.28(1.09) &   0.27(0.07) &   0.06(0.03) \\ 
			Lap & rate & 100.00(0.00) & 100.00(0.00) &  27.67(30.72) &  52.00(45.52) & 100.00(0.00) & 100.00(0.00) \\ 
			& size &   5.11(4.34) &   5.07(4.43) &   2.78(4.39) &   5.55(6.93) &   4.41(2.79) &   3.00(0.00) \\ 
			& loss &   0.07(0.04) &   0.08(0.05) &   1.95(0.50) &   1.25(1.09) &   0.26(0.06) &   0.06(0.02) \\ 
			\hline \\[-1.8ex] 
	\end{tabular} }
\end{table} 

\section{AN APPLICATION}
In this section, we apply our proposal to the red wine   dataset which is publicly  available at \url{https://archive.ics.uci.edu/ml/datasets/Wine+Quality}. The data consist of 11 measurements of  several chemical constituents, including determination of density, alcohol or pH values for 1599 red wine samples from the northwest region of Portugal. The response variable is the median of the scores evaluated by human experts and each score ranges from 0 (very bad) to 10 (very excellent). The same dataset was once analyzed by \cite{wine}. In their analysis,  interactions are found to be very helpful for  prediction. The original data are 1599 observations on 11 covariates. To accommodate high dimensional setting,   we follow \cite{radchenko2010variable} and standardize all the variables and conduct the following two experiments: 
\begin{itemize}
	\item {\bf Experiment 1.} Denote the original 11 covariates as $X_1,\ldots, X_{11}$. We add 100 noise variables $X_{12},\ldots,X_{111}$ to the data, where $X_{12},\ldots, X_{61}$ are generated from the standard normal distribution and the remainders are generated by the uniform distribution on the interval $[-\sqrt 3, \sqrt 3]$.  
	\item {\bf Experiment 2.} We generate  the  covariates in the same way as in Experiment 1. In addition, we modify the response variable $Y$ by adding two more interactions: $Y+0.5 X_{12}X_{13}+0.5 X_{61}X_{62}$. In this experiment, both the strong and the weak  heredity conditions are violated.
\end{itemize}  
In both experiments  the covariate dimension $p=111$, leading to $111\times 100/2=6,105$ possible interactions. We randomly select 400 observations as the sample and the procedure is repeated 100 times. The heat map of the frequencies of the identified interactions are summarized in Figure \ref{figreal}. It can be clearly seen that, in Experiment 1, the detected interactions mainly occur among the first 11 covariates collected in the original dataset while the interactions related to the remaining 100 noisy covariates are rarely   detected. This indicates that both PIEy and PIEr are able to exclude irrelevant interactions. In Experiment 2, both methods are able to exclude irrelevant interactions with high probability. In addition,   the interactions  $X_{12}X_{13}$  and $X_{61}X_{62}$ are successfully detected throughout. 

\begin{figure}[!ht]
	\centerline{
		\begin{tabular}{cc}				
			\psfig{figure=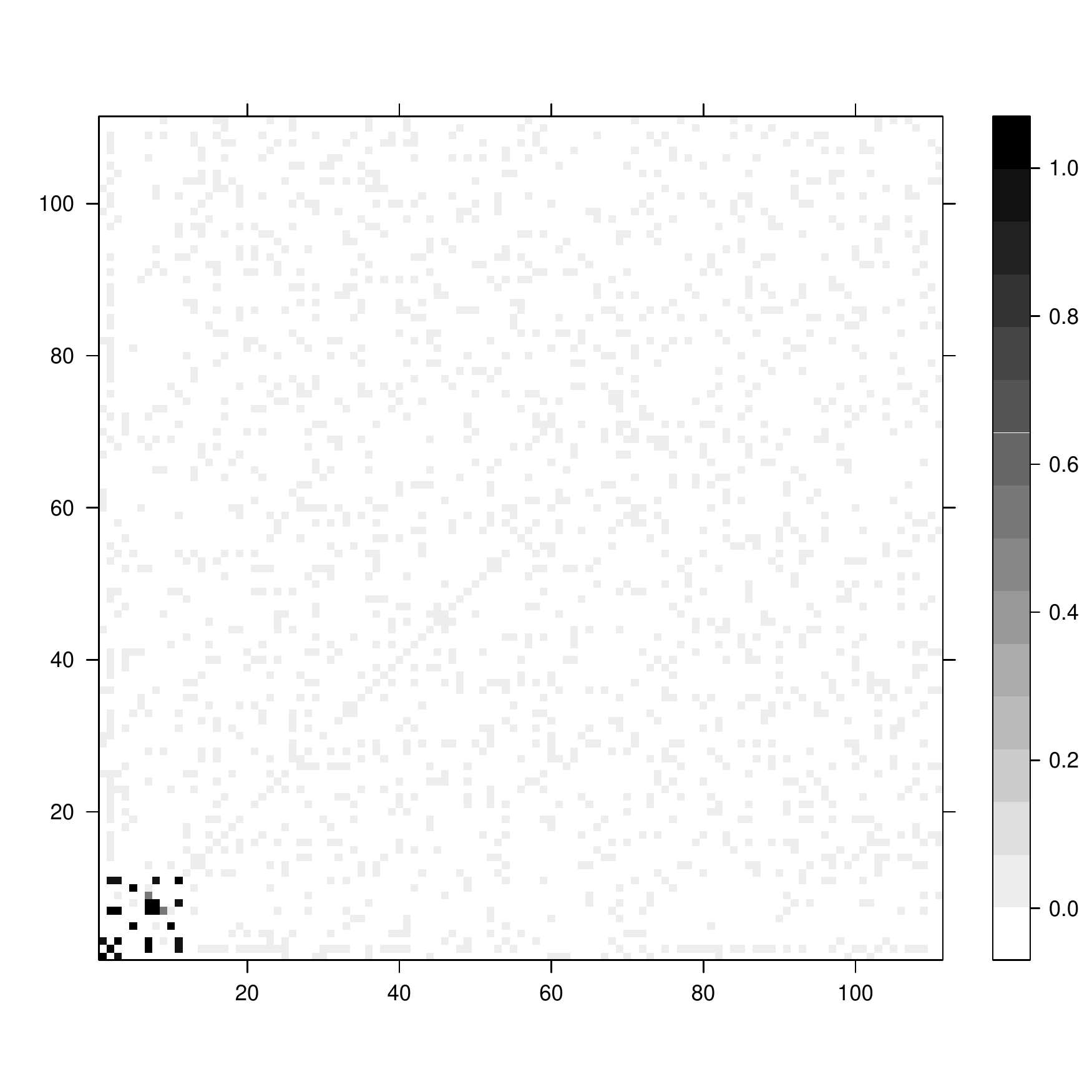,width=2.5 in,height=2.5 in,angle=0} &
			\psfig{figure=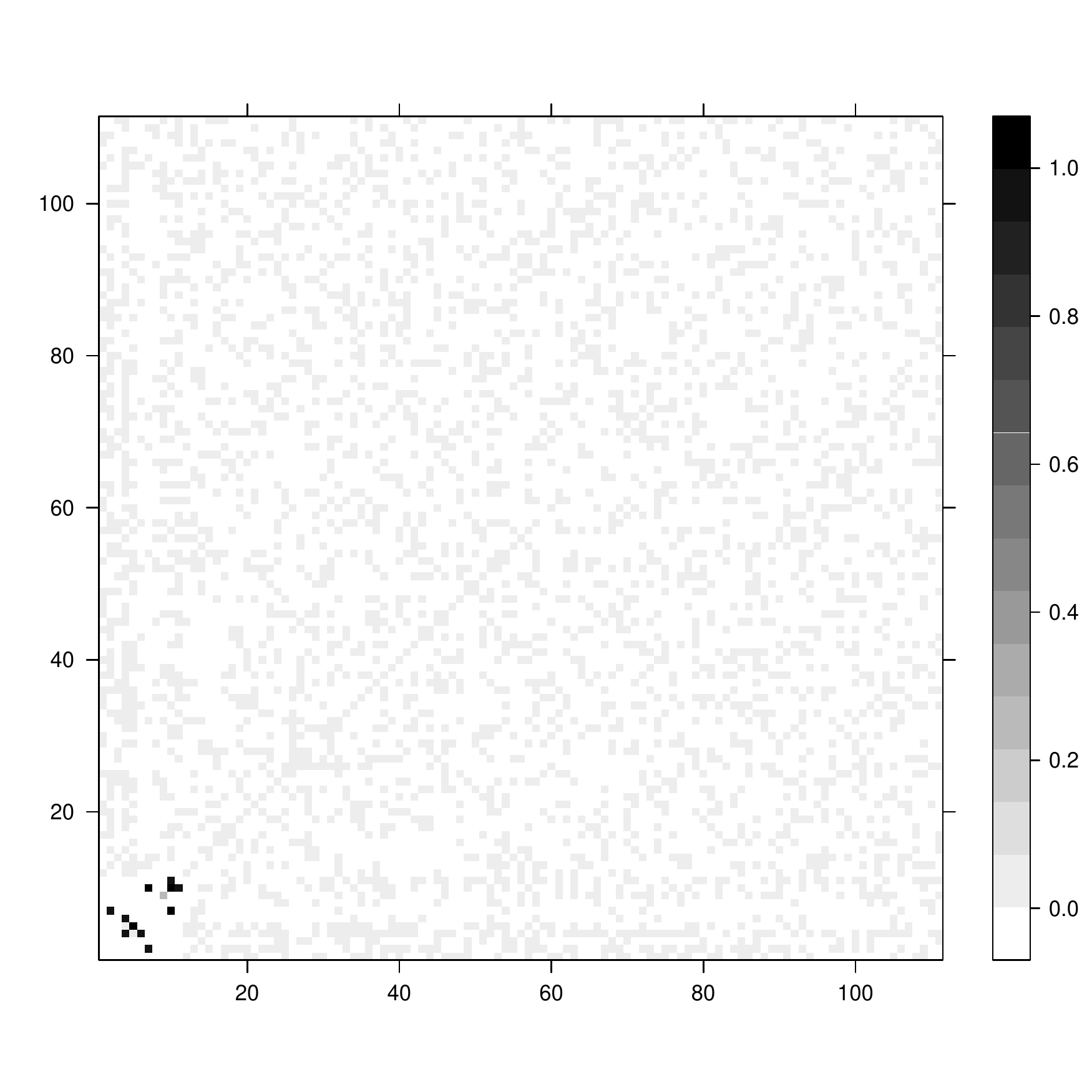,width=2.5 in,height=2.5 in,angle=0} \\
			(a) PIEy
			& (b) PIEr \\
			\psfig{figure=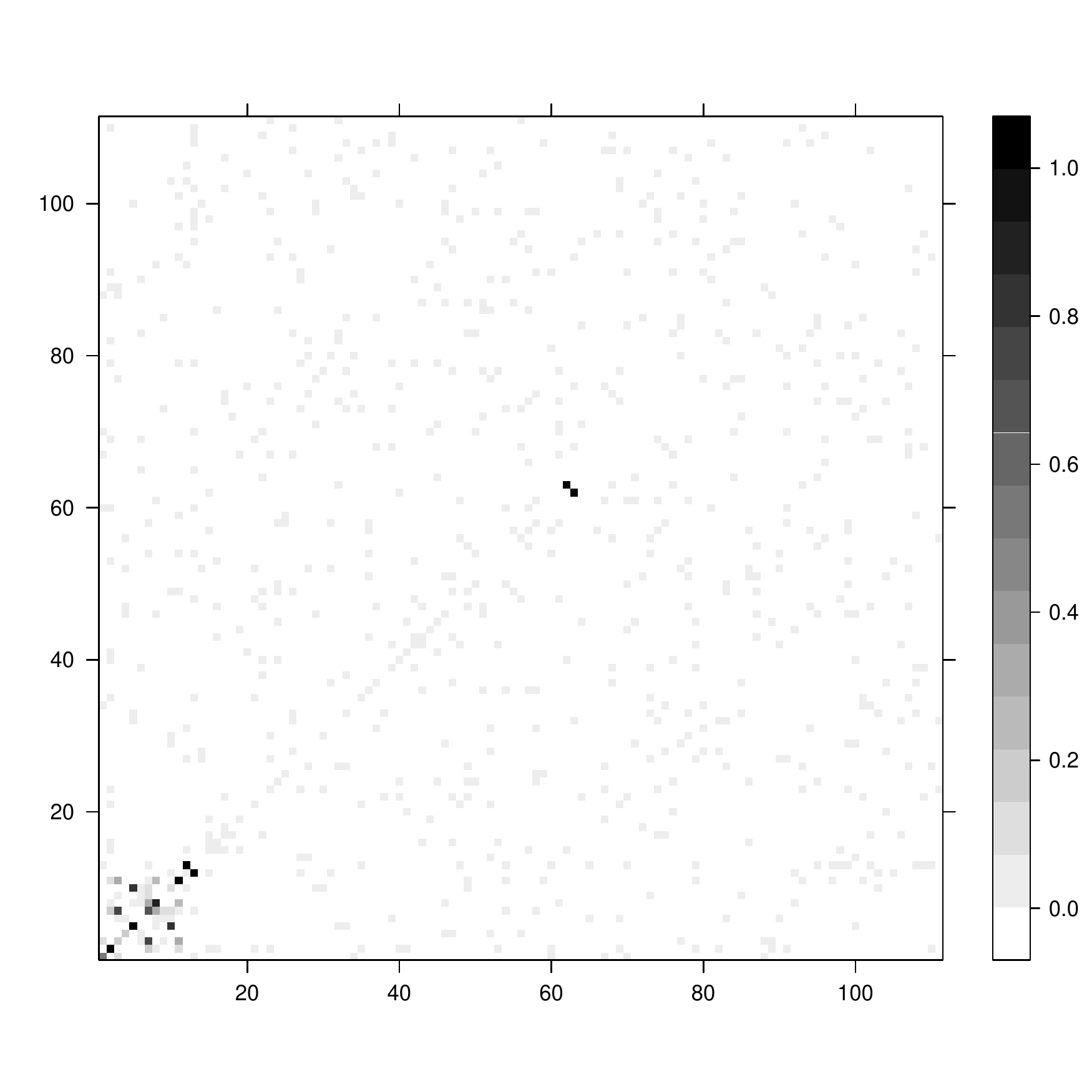,width=2.5 in,height=2.5 in,angle=0} &
			\psfig{figure=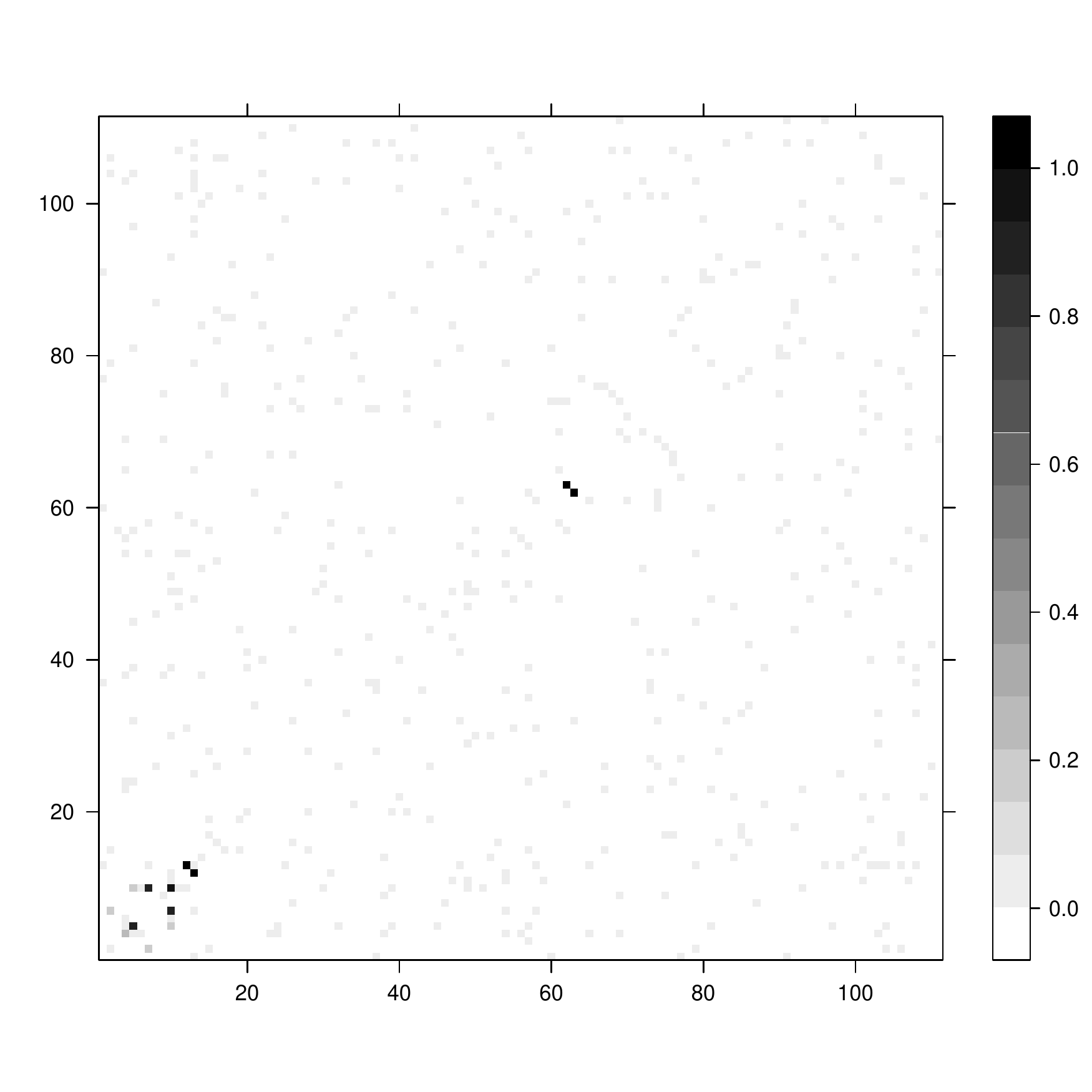,width=2.5 in,height=2.5 in,angle=0} \\
			(c) PIEy
			& (d) PIEr
		\end{tabular}
	}
	\caption{\textit{Heat maps of frequency of the interactions identified out of 100 replications using PIEy and PIEr. Upper panel: Experiment 1. Lower panel: Experiment 2.}}
	\label{figreal}
\end{figure}

\section{DISCUSSION}
In this paper we propose a penalized  estimation to detect interactions without requiring heredity conditions. We develop an efficient ADMM algorithm to implement our estimation. We demonstrate the effectiveness of our proposal through extensive numerical studies. We remark here that, if the strong or the weak condition is satisfied, some existing methods, such as the RAMP method, work pretty well. However, when we have little prior information about whether the heredity condition holds true or not in an application, we advocate using our proposal in that it does not require this assumption. In effect, if the heredity condition is known to be satisfied, we can also incorporate it into our proposal through a two-stage procedure. In the first stage, we use the penalized least squares to identify the main effects; and in the second stage, we implement our procedure using only the main effects that are selected in the first stage. This allows us to handle ultrahigh dimensional problems efficiently. Another way to enhance the power of our proposal is to incorporate some screening procedures into our problems.  We also remark here that, in the present context we focus on quadratic regression which contains pairwise  interactions of the form $(X_{k_1}X_{k_2})$. We remark here that our idea can be  generalized  naturally  to higher-order interactions models of the form $(X_{k_1}X_{k_2}\cdots X_{k_d})$ $(d>2)$.  However, estimating high-order interactions is generally much more challenging because there are {\large$\binom{p}{d}$}  possible interactions of order $d$ in total. The central task is possibly  to  develop efficient algorithms with minimal computational complexity. Researches along these lines are warranted.

\section*{Appendix}
\subsection{Appendix A: Some Useful Lemmas}
We first show that the ADMM algorithm to minimize  \eqref{admm0} converges linearly.  

\begin{lemm}\label{linearADM}
	Given $\hSig$ and $\hLam$. Suppose that the ADMM algorithm \eqref{AlgoB}-\eqref{Algo:L} generates a   sequence of solutions $\{(\B^k,\bPsi^k, \L^k), k=1,\ldots\}$. Then  $\{\B^k,\bPsi^k\}$ converges linearly to the minimizer of \eqref{admm0}, and $\|\B^k-\bPsi^k\|_F$ converges linearly to zero. 
\end{lemm}
\begin{proof}
	The objective function in the minimization problem    \eqref{admm0}  can be decomposed into two components: $f(\B,\bPsi)=f_1(\B)+f_2(\bPsi)$, where
	$f_1(\B)\defby \tr \{(\B \hSig)^2\}-\tr (\B \hLam)$ and 
	$f_2(\bPsi)\defby\lambda_n \| \bPsi\|_1$. Rewrite $\tr \{(\B \hSig)^2\}=\vec(\B)\trans(\hSig\otimes \hSig)\vec(\B)$. Denote $\hSig\otimes \hSig=\U\trans\bLam \U$ and $\A_1=\U\trans\bLam^{1/2} \U$. Let  $g_1(\x)\defby\|\x\|_F^2$ be a function  defined on $\mR^{p^2}\mapsto \mR$, and $h_1(\x)\defby\tr(\hLam\x)$, $h_2(\x)\defby\lambda_n\|\x\|_1$ be two functions defined on $\mR^{p^2}\mapsto \mR$. Then $f_1(\B)=g_1\{\A_1\vec(\B)\}+h_1\{\vec(\B)\}$ and $f_2(\bPsi)=h_2\{\vec(\bPsi)\}$. 
	Given $\hSig$, $\hLam$ and $\lambda_n$,  the gradient of $g_1$ is  uniformly Lipschitz continuous and $h_1$ and $h_2$ are  polyhedral. Lemma  \ref{linearADM} thus follows immediately from Theorem 3.1 of \cite{hong2017linear}. 
\end{proof}

Next we present some useful lemmas for the proofs of the main theorems. Without loss of generality, in what follows we assume that $\E(\x)={\bf0}$ and $\E (Y)=0$.
{\lemm\label{lem01} Let $W_1,\cdots,W_n$ be independent variables and ${\rm E}\{\exp(c_1 |W_i|^{\alpha_0})\} < A_0$ for some $0<\alpha_0 \leq 1, c_1>0, , A_0>0$. Then for $0<t\leq 1$, there exist constants $c_2,c_3>0$ such that
	\beqrs
	\pr\Big\{\Big| n^{-1} \sum_{i=1}^n (W_i- {\rm E} W_i) \Big|>t\Big\} \leq c_2 \exp(-c_3 n^{\alpha_0} t^2).
	\eeqrs
}
\noindent\textit{Proof of Lemma \ref{lem01}:}  For $\E W_i=0$, see Lemma B.4 of \cite{hao2014interaction}. Here, we only need to show $\E\{\exp(c_1 |W_i-\E W_i|^{\alpha_0})\} < A_1$ for some $A_1>0$. By the integral identity of the expectation, we have
\beqrs
\E|W_i| 
&=&\int_{0}^{\infty} \pr \left\{ \exp(c_1 |W_i|^{\alpha_0})>\exp(c_1t^{\alpha_0})\right\}  dt\\
& \leq & \int_{0}^{\infty} \E \{\exp(c_1 |W_i|^{\alpha_0})\} \exp(-c_1t^{\alpha_0}) dt
\leq  A_0  \int_{0}^{\infty} \exp(-c_1t^{\alpha_0}) dt \defby c.
\eeqrs
Consequently,
$\E\{\exp(c_1 |W_i-\E W_i|^{\alpha_0})\}  \leq  \E\{\exp(2c_1 |W_i|^{\alpha_0}+2c_1 |\E W_i|^{\alpha_0})\} \leq  A_0 \exp(2c_1 c^{\alpha_0})\defby A_1.$
The proof is completed. \hfill$\fbox{}$

{\lemm \label{lem02} Let $W_1$ and $W_2$ be two variables such that
	${\rm E} \{ \exp(c_1 |W_1|^{\alpha_1})\}\le A_1$ and ${\rm E}\{\exp(c_2 |W_2|^{\alpha_2})\}\le A_2$, 
	where $c_1,c_2,\alpha_1,\alpha_2,A_1,A_2>0$. We have
	\beqrs
	{\rm E} \big\{\exp\big(\min(c_1,c_2) |W_1 W_2|^{ {\alpha_1 \alpha_2}/{(\alpha_1+\alpha_2)}}\big)\big\}<\max(A_1,A_2).  
	\eeqrs
}
\noindent\textit{Proof of Lemma \ref{lem02}:} By Holder's or Young's inequality, 
\beqrs
&&\E \big\{\exp\big(\min(c_1,c_2) |W_1 W_2|^{ {\alpha_1 \alpha_2}/{(\alpha_1+\alpha_2)}}\big)\big\} \\
&\leq & \min(c_1,c_2) \E\big[ \exp\{  |W_1|^{\alpha_1}{\alpha_2}/{(\alpha_1+\alpha_2)}  +  |W_2|^{\alpha_2}{\alpha_1}/{(\alpha_1+\alpha_2)}\}\big]\\
&\leq &    {A_1\alpha_2}/{(\alpha_1+\alpha_2)}   + A_2{\alpha_1}{(\alpha_1+\alpha_2)}\leq  \max(A_1,A_2).
\eeqrs
The proof is completed.  \hfill$\fbox{}$

{\lemm\label{lem03}
	Under condition (A2), we have there exists a constant $C>0$,
	\beqr \label{lem31}
	&&\pr \left(  \|\bar{\x} \bar{\x} \trans\|_{\infty}\geq C \log(p)/n \right)=O(p^{-1}), ~~~{\rm and}
	\\ \label{lem32}
	&& \pr\{ \|n^{-1} \sum_{i=1}^n \x_i \x_i \trans
	-\bSig\|_{\infty} \geq C \{\log(p)/{n}\}^{1/2}\} =O(p^{-1}).
	\eeqr}
\noindent\textit{Proof of Lemma \ref{lem03}:} Writing $\e_k$ as the unit-length $p$-vector with its $k$-th entry being one, we have
$
\|\bar{\x} \bar{\x} \trans\|_\infty = \underset{k,l}\max  |\e_k \trans \bar{\x} \bar{\x} \trans \e_l|.
$
Note that $\e_k \trans \x_1,\cdots, \e_k \trans \x_n$ are independent centered sub-Gaussian variables. By Hoeffding's inequality \citep[Theorem 2.6.3]{vershynin2016high},
$\pr(|\e_k \trans \bar{\x}|\geq t) \leq 2 \exp(-c n t^2),\textrm{ for any  } t \ge 0,$
and then
$\pr(|\e_k \trans \bar{\x} \bar{\x} \trans \e_l|\geq t) \leq \pr(|\e_k \trans \bar{\x}|\geq \sqrt{t})+\pr(|\e_l \trans \bar{\x}|\geq \sqrt{t})\leq 4 \exp(-c n t).
$
Therefore,
\beqrs
\pr \left\{ \|\bar{\x} \bar{\x} \trans\|_{\infty} \geq t \right\} \leq \sum_{k,l} \pr(|\e_k \trans \bar{\x} \bar{\x} \trans \e_l|\geq t)\leq 4p^2 \cdot \exp(-cn t).
\eeqrs
Set $t= c^{-1}C \log{(p)}/n$ for large enough $C$, which yields the conclusion \eqref{lem31}.

Similarly, $\e_k \trans \x_i \x_i \trans \e_l-\e_k \trans\bSig \e_l,~i=1,\cdots,n$ are independent centered sub-exponential variables. By Bernstein's inequality \citep[Theorem 2.8.2]{vershynin2016high}, we get
\beqrs
&&\pr\left\{ \Big|\e_k \trans\Big(n^{-1} \sum_{i=1}^n \x_i \x_i \trans-\bSig\Big) \e_l\Big|\geq t \right\} \leq  2 \exp\left\{-n  \min(c_1 t^2,c_2t) \right\}, ~~{\rm and}
\\
&&\pr\left( \Big\|n^{-1} \sum_{i=1}^n \x_i \x_i \trans
-\Sigma\Big\|_{\infty} \geq t\right)\leq 2 p^2 \cdot \exp\left\{-n  \min(c_1 t^2,c_2t) \right\}.
\eeqrs
Choose $t= C\{\log{(p)}/n\}^{1/2}$ with a sufficiently large $C$ to complete  proof of \eqref{lem32}.    \hfill$\fbox{}$

{\lemm\label{lem04}
	Under conditions (A2) and (A3), there exists a constant $C>0$ such that,
	\beqr  
	&& \\ \label{lem41}
&& \pr \left[  \Big\|n^{-1}\sum_{i=1}^n Y_i \x_i -{\rm E} Y\x\Big\|_\infty \geq C  \{n^{-\alpha/(\alpha+1)}\log(p)\}^{1/2} \right]	
	=O(p^{-1}), ~{\rm and} \nonumber\\
	&& \\ \label{lem42}
	&& \pr\left[ \Big\|n^{-1}\sum_{i=1}^n Y_i \x_i \x_i \trans-{\rm E} Y \x \x \trans\Big\|_\infty\geq C  \{n^{-\alpha/(\alpha+1)}\log(p)\}^{1/2} \right]=O(p^{-1}).\nonumber
	\eeqr
}
\noindent\textit{Proof of Lemma \ref{lem04}:} We prove \eqref{lem42} only in what follows and \eqref{lem41} can be proved using similar arguments. For $\e_k, \e_l$,
\begin{align*}
\e_k \trans \Big(n^{-1}\sum_{i=1}^n Y_i \x_i \x_i \trans-{\rm E} Y \x \x \trans\Big) \e_l=n^{-1} \sum_{i=1}^n Y_i (\e_k \trans \x_i) (\e_l \trans \x_i)-\e_k \trans ({\rm E} Y \x \x \trans) \e_l.
\end{align*}
By condition (A2), there exist constants $c_0$ and $C_0$ such that
\[ {\rm E} \{\exp(c_0 |\e_k \trans \x_i \e_l \trans \x_i| )\} \leq {\rm E} \{\exp ( c_0 |\e_k \trans \x_i|^2 )\}+{\rm E} \{\exp ( c_0 |\e_l \trans \x_i|^2)\}\leq 2C_0.
\]
By condition (A3) and Lemma \ref{lem02}, we have there exist constants $c_2, C_2$ such that
$
{\rm E} \Big\{\exp \Big(c_2 |Y_i (\e_k \trans \x_i) (\e_l \trans \x_i)|^{{\alpha}/{(\alpha+1)}}  \Big)\Big\}\le C_2.
$
By Lemma \ref{lem01}, we have
\beqrs
\pr \Big\{ \Big| \e \trans (n^{-1}\sum_{i=1}^n Y_i \x_i \x_i \trans-{\rm E} Y \x \x \trans) \trans \tilde{\e} \Big| \geq t \Big\} \leq   c_2 \exp(-c_3 n^{ {\alpha}/({\alpha+1})} t^2).
\eeqrs
Using the similar arguments as in the proofs of Lemma \ref{lem03}, we can show
\beqrs
\pr\Big[ \Big\|n^{-1}\sum_{i=1}^n Y_i \x_i \x_i \trans-\E Y \x \x \trans\Big\|_\infty\geq C  \{n^{-\alpha/(\alpha+1)}\log(p)\}^{1/2} \Big]=O(p^{-1}).
\eeqrs
%
The proof is completed.  \hfill$\fbox{}$

\subsection{Appendix B: The $\ell_1$-Penalized Estimation}
Let $\A \in \mR^{q \times q},~\a \in \mR^q$ be unknown parameters and $\A$ is a positive definite symmetric matrix. 
To estimate $\b^\ast \defby \A^{-1} \a$, we consider the $\ell_1$-penalized approach:  
\beqr\label{lemmal1}
\wh{\b}=\argmin_{\b \in \mR^q}  \b \trans \wh{\A} \b/2-\wh{\a} \trans \b+\lambda \|\b\|_1,
\eeqr
where $\lambda$ is the tuning parameter and $\hat{\A}$ and $\hat{\a}$ are the empirical estimators of $\A$ and $\a$, respectively.
In the sequel, we establish theoretical results for solving \eqref{lemmal1}. These general results will then be used to prove the main theorems in our paper.
{\lemm\label{lem05}  Denote $\Delta=\|\wh{\a}-\a \|_{\infty}+\|(\wh{\A}-\A)\b^\ast\|_{\infty}$ and let $\supp=\{i: \b^\ast_i\neq 0 \}$ be the support of $\b^\ast$.
	Assume that
	$	\|\A_{\supp^c,\supp}\A_{\supp,\supp}^{-1}\|_L+2\|\b^\ast \|_0 \|\A_{\supp,\supp}^{-1}\|_{L} \|\wh{\A}-\A\|_{\infty}<1,
	$
	and
	$
	\lambda>2  ( 1-\|\A_{\supp^c,\supp}\A_{\supp,\supp}^{-1}\|_L-2\|\b^\ast \|_0 \|\A_{\supp,\supp}^{-1}\|_{L} \|\wh{\A}-\A\|_{\infty} )^{-1}\Delta,
	$
	we have 
	\begin{itemize}
		\item[(i)] $\wh \b_{\supp^c}={\bf 0}$;
		\item[(ii)] 
		$\|\wh{\b}-\b^\ast\|_\infty \leq {2\lambda }({1-\|\b^\ast \|_0 \|\A_{\supp,\supp}^{-1}\|_{L} \|\wh{\A}-\A\|_{\infty}})^{-1}\|\A_{\supp,\supp}^{-1}\|_{L}.$
	\end{itemize}
}

\noindent\textit{Proof of Lemma \ref{lem05}:}
Given the true support $\supp$,  we consider the estimation
\beqrs
\wh{\b}^0&=&\argmin_{\b \in \mR^q,~\b_{\supp^c}=0}  \b \trans \wh{\A} \b/2-\wh{\a} \trans \b+\lambda \|\b\|_1\\
&=&\argmin_{\b \in \mR^q,~\b_{\supp^c}=0}   \b_{\supp} \trans \wh{\A}_{\supp,\supp} \b_{\supp}/2-\wh{\a}_{\supp} \trans \b_{\supp}+\lambda \|\b_{\supp}\|_1.
\eeqrs
By the Karush-Kuhn-Tucker (KKT) condition, we have
\beqr \label{la6}
\wh{\A}_{\supp,\supp} \wh{\b}^0_{\supp}-\wh{\a}_{\supp} =-\lambda \Z,
\eeqr
where $\Z$ is the sub-gradient of $\|\b _{\supp}\|_1$. 
By the definition of $\b^\ast=\A^{-1} \a$, we have
\beqrs
\begin{pmatrix}
	\a_{\supp}\\
	\a_{\supp^c}
\end{pmatrix}=\begin{pmatrix}
	\A_{\supp,\supp}& \A_{\supp,\supp^c}\\
	\A_{\supp^c,\supp}& \A_{\supp^c,\supp^c}\\
\end{pmatrix} \begin{pmatrix}
	\b^\ast_{\supp}\\
	\textbf{0}\\
\end{pmatrix}=\begin{pmatrix}
	\A_{\supp,\supp} \b^\ast_{\supp}\\
	\A_{\supp^c,\supp} \b^\ast_{\supp}
\end{pmatrix},
\eeqrs
and hence we have
$
\wh{\A}_{\supp,\supp} \wh{\b}^0_{\supp}-\A_{\supp,\supp} \b^\ast_{\supp}+\a_{\supp}-\wh{\a}_{\supp}=-\lambda \Z.
$
Consequently, we obtain,
\beqr\label{la33}
\wh{\b}^0_{\supp}-\b^\ast_{\supp}=-\A_{\supp,\supp}^{-1} \left\{ \lambda \Z+(\wh{\A}_{\supp,\supp}-\A_{\supp,\supp})\wh{\b}^0_{\supp}+(\a_{\supp}-\wh{\a}_{\supp})\right\}.
\eeqr
Using the triangle inequality, we can show that,
\beqrs
&&\|\wh{\b}^0_{\supp}-\b^\ast_{\supp}\|_\infty   \\
&\leq& \|\A_{\supp,\supp}^{-1}\|_{L}\Big\{\lambda  \|\Z\|_{\infty}+ \|(\wh{\A}_{\supp,\supp}-\A_{\supp,\supp})(\wh{\b}^0_{\supp}-\b^\ast_{\supp})\|_\infty \\
&&+\|(\wh{\A}_{\supp,\supp}-\A_{\supp,\supp})\b_{\supp}^\ast+\a_{\supp}-\wh{\a}_{\supp} \|_{\infty}\Big\}\\
&\leq&   \|\A_{\supp,\supp}^{-1}\|_{L} \Big\{ \lambda+\|\b^{\ast}\|_0 \|\wh{\A}-\A\|_{\infty} \|\wh{\b}^0_{\supp}-\b^\ast_{\supp}\|_\infty+\|(\wh{\A}-\A)\b^\ast+\a-\wh{\a} \|_{\infty}\Big\},
\eeqrs
which implies that
\beqr \label{la3}
&&\|\wh{\b}^0_{\supp}-\b^\ast_{\supp}\|_\infty \\
&\leq&  (1-\|\b^{\ast}\|_0 \|\A_{\supp,\supp}^{-1}\|_{L} \|\wh{\A}-\A\|_{\infty})^{-1} \|\A_{\supp,\supp}^{-1}\|_{L}(\lambda+\Delta). \nonumber
\eeqr
Next, we  show that
$\wh{\b}^0$ is exactly the minimizer to $\min\limits_{\b \in \mR^q}   \b \trans \wh{\A} \b/2-\wh{\a} \trans \b+\lambda \|\b\|_1. $
By the KKT condition, it is sufficient to prove
\beqr
&& \|(\wh\A \wh{\b}^0-\wh{\a})_{\supp}\|_{\infty} \leq \lambda, \textrm{ and }\label{la1}\\
&& \|(\wh\A \wh{\b}^0-\wh{\a})_{\supp^c}\|_{\infty} <\lambda. \label{la2}
\eeqr
Since $(\wh\A \wh{\b}^0-\wh{\a})_{\supp}=\wh{\A}_{\supp,\supp} \wh{\b}^0_{\supp}-\wh{\a}_{\supp}$, \eqref{la1} is true by \eqref{la6}. For \eqref{la2}, we have
\beqrs
(\wh\A \wh{\b}^0-\wh{\a})_{\supp^c}&=&\wh{\A}_{\supp^c,\supp} \wh{\b}^0_{\supp}-\wh{\a}_{\supp^c} 
= \wh{\A}_{\supp^c,\supp} \wh{\b}^0_{\supp}-\A_{\supp^c,\supp} \b^\ast_{\supp}+\a_{\supp^c}-\wh{\a}_{\supp^c}\\
&=&\wh{\A}_{\supp^c,\supp} (\wh{\b}^0_{\supp}-\b^\ast_{\supp})+(\wh{\A}_{\supp^c,\supp}-\A_{\supp^c,\supp})\b^\ast_{\supp}+\a_{\supp^c}-\wh{\a}_{\supp^c}\\
&=&(\wh{\A}_{\supp^c,\supp}-\A_{\supp^c,\supp})(\wh{\b}^0_{\supp}-\b^\ast_{\supp})+\A_{\supp^c,\supp}\A_{\supp,\supp}^{-1} \{\A_{\supp,\supp}(\wh{\b}^0_{\supp}-\b^\ast_{\supp})\}\\
&&+\{(\wh{\A}-\A)\b^\ast+\a-\wh{\a}\}_{\supp^c}.
\eeqrs
Thus, it follows from \eqref{la33} and \eqref{la3} that  $\|(\wh\A \wh{\b}^0-\wh{\a})_{\supp^c}\|_{\infty}$ is less than or equal to
\beqrs
& & \|\b^{\ast}\|_0 \|\wh \A-\A\|_{\infty} \|\wh{\b}^0_{\supp}-\b^\ast_{\supp}\|_{\infty}\\
&& +\|\A_{\supp^c,\supp}\A_{\supp,\supp}^{-1}\|_L(\lambda+\Delta+ \|\b^{\ast}\|_0 \|\wh \A-\A\|_{\infty} \|\wh{\b}^0_{\supp}-\b^\ast_{\supp}\|_{\infty}) +\Delta\\
& \leq & \frac{ (1+ |\A_{\supp^c,\supp}\A_{\supp,\supp}^{-1}\|_L)(\lambda+\Delta)}{1-\|\b^{\ast}\|_0 \|\A_{\supp,\supp}^{-1}\|_{L} \|\wh{\A}-\A\|_{\infty}}-\lambda\\
&=& \lambda+\Bigg\{\Delta-\frac{1- |\A_{\supp^c,\supp}\A_{\supp,\supp}^{-1}\|_L-2\|\b^{\ast}\|_0 \|\A_{\supp,\supp}^{-1}\|_{L} \|\wh{\A}-\A\|_{\infty}}{1+ |\A_{\supp^c,\supp}\A_{\supp,\supp}^{-1}\|_L} \lambda\Bigg\}\\
&&\hspace{.8cm}\left\{ \frac{1+ \|\A_{\supp^c,\supp}\A_{\supp,\supp}^{-1}\|_L}{1-\|\b^{\ast}\|_0 \|\A_{\supp,\supp}^{-1}\|_{L} \|\wh{\A}-\A\|_{\infty}}\right\}.
\eeqrs
When $\lambda> {2({1-\|\A_{\supp^c,\supp}\A_{\supp,\supp}^{-1}\|_L-2\|\b^\ast \|_0 \|\A_{\supp,\supp}^{-1}\|_{L} \|\wh{\A}-\A\|_{\infty} })^{-1}\Delta} ,$
we have
$\|(\wh\A \wh{\b}^0-\wh{\a})_{\supp^c}\|_{\infty}<\lambda.
$
Consequently, $\wh{\b}=\wh{\b}^0$ and \eqref{la2} is an immediate result of \eqref{la3} by noting $\Delta\leq \lambda$. The proof is completed.  \hfill$\fbox{}$

\subsection{Appendix C: Proof of Proposition \ref{prop1} }
Recall that $\E(Y\mid\x) = \alpha + (\x-\u) \trans \bb + (\x-\u) \trans \bOme (\x-\u)$. Direct calculations  show
\beqrs
\cov(\x,Y)&=&E\Big[\left\{Y-E(Y)\right\}(\x-\u)\Big]\\
&=& E\Big[\left\{(\x-\u) \trans \bb + (\x-\u) \trans \bOme (\x-\u)-\tr(\bOme \bSig)\right\}(\x-\u)\Big]\\
&=& E \left\{(\x-\u)(\x-\u) \trans \bb+(\z \trans \bGam_0 \trans \bOme \bGam_0 \z) \bGam_0 \z\right\}=\bSig \bb.
\eeqrs
The proof of the first part is completed. Next we prove the second part.
\beqrs
\bLam_y 
&=& E\Big[\left\{(\x-\u) \trans \bb + (\x-\u) \trans \bOme (\x-\u)-\tr(\bOme \bSig)\right\}(\x-\u)(\x-\u)\trans\Big]\\
&=&E(\x-\u) (\x-\u)  \trans \bOme (\x-\u)(\x-\u)\trans- \tr(\bOme \bSig) \bSig\\
&=& E \left\{  \bGam_0 \z \z \trans (\bGam_0 \trans \bOme \bGam_0) \z \z \trans \bGam_0 \trans\right\}- \tr (\bGam_0 \trans \bOme \bGam_0) \bGam_0 \bGam_0 \trans\\
&=& \bGam_0 \Big[E \left\{  \z \z \trans (\bGam_0 \trans \bOme \bGam_0) \z \z \trans\right\}- \tr (\bGam_0 \trans \bOme \bGam_0) \I_p \Big]\bGam_0 \trans \\
&=&  \bGam_0 \{2 \bGam_0 \trans \bOme \bGam_0-(\Delta-3) \diag(\bGam_0 \trans \bOme \bGam_0)\} \bGam_0 \trans  \\
&=& 2 \bSig \bOme \bSig-(\Delta-3) \bGam_0 \diag(\bGam_0 \trans \bOme \bGam_0) \bGam_0 \trans.
\eeqrs		
Thus,  $\bOme=  \bSig^{-1} \bLam \bSig^{-1}\big/2$ when $\Delta=3$ or  $\diag(\bGam_0 \trans \bOme \bGam_0)=0$. The proof is completed.  \hfill$\fbox{}$

\subsection{Appendix D: Proof of Theorem \ref{thm1}}
We provide proofs for (i) and (iii) in what follows because  (ii) is an immediate result of (i) and (iii)  and (iv) can be obtained analog to (iii).
For the target parameter matrix $2\bOme=\bSig^{-1}\bLam \bSig^{-1}$, we consider its vectorization
\beqr
2 \vec{(\bOme)}=\vec{(\bSig^{-1}\bLam \bSig^{-1})}=(\bSig^{-1}\otimes \bSig^{-1}) \vec{(\bLam)}=\bGam^{-1}  \vec{(\bLam)},
\eeqr
where $\bGam=\bSig \otimes \bSig$ is a positive and symmetric matrix. For the estimation,
\beqrs
\hOme_y=\argmin_{\B \in \mR^{p \times p}} \tr \{(\B\hSig)^2\}-\tr (\B \hLam_y)+\lambda_{1n} \|\B\|_1.
\eeqrs
Equivalently, we have
\beqrs
\vec{(\hOme_y)}=\argmin_{\B \in \mR^{p \times p}} \vec{(\B)}  \trans \hGam \vec{(\B)} -\vec{(\hLam_y)} \trans \vec{(\B)}+\lambda_{1n} \|\vec{(\B)}\|_1,
\eeqrs
where
$ \hGam \defby \hSig \otimes \hSig$. 
Therefore, we can use Lemma \ref{lem05} to derive the theoretical properties by letting
$
\A=2 \bGam,~\a=\vec{(\bLam_y)},~\wh \A=2\hGam~\mbox{and}~\wh \a=\vec{(\hLam_y)}.
$

Recall the definition of $\hSig$ and $\hLam_y$.
\beqrs
\hSig&=&n^{-1}\sum_{i=1}^n(\x_i-\bar{\x})(\x_i-\bar{\x}) \trans=n^{-1}\sum_{i=1}^n \x_i \x_i \trans-\bar{\x} \bar{\x} \trans, ~~~{\rm and} \\
\hLam_y 
&=&n^{-1}\sum_{i=1}^n Y_i \x_i \x_i \trans-n^{-1}\sum_{i=1}^n Y_i (\bar{\x} \x_i \trans +\x_i \bar{\x} \trans)-n^{-1} \bar{Y}\sum_{i=1}^n \x_i \x_i \trans+2 \bar{Y} \bar{\x} \bar{\x} \trans.
\eeqrs
Lemmas \ref{lem03} and \ref{lem04} ensure that there exists a constant $C>0$ such that with probability greater than $1-O(p^{-1})$,
$\|\hSig
-\bSig\|_{\infty} \leq C (n^{-1}\log p)^{1/2}$ and $\|\hLam_y-\bLam_y\|_\infty \leq C \{n^{-\alpha/({\alpha+1})}\log{(p)}\}^{1/2}.$
Note that
$\|\hGam-\bGam\|_\infty=\|\hSig \otimes(\hSig-\bSig)+(\hSig-\bSig)\otimes \bSig\|_\infty\leq (\|\hSig\|_\infty+\|\bSig\|_\infty) \|\hSig-\bSig\|_\infty,$
with probability greater than $1-O(p^{-1})$, we have,
$\|\hGam
-\bGam\|_{\infty} \leq C_1 (n^{-1}\log p)^{1/2},$
for some constant $C_1>0$ and
$\|\bGam_{\supp^c,\supp}\bGam_{\supp,\supp}^{-1}\|_L+2\|\bOme \|_0 \|\bGam_{\supp,\supp}^{-1}\|_{L} \|\hGam-\bGam\|_{\infty} \leq  1-\kappa+2C_1 Ms_p (n^{-1}\log p)^{1/2} = 1-\kappa+o(1)<1.$
Next, we consider 
$
\Delta_1 \defby \|\vec{(\hLam_y)}-\vec{(\bLam_y)}\|_\infty +2\|(\hGam-\bGam) \vec{(\bOme)}\|_{\infty}. 
$
Note that
\beqrs
\|(\hGam-\bGam) \vec{(\bOme)}\|_{\infty}&=&\|(\hSig \otimes \hSig-\bSig \otimes \bSig) \vec{(\bOme)} \|_{\infty}\\
&=&\|\vec{(\hSig \bOme \hSig-\bSig \bOme \bSig)}\|_{\infty} 
=  \|\hSig \bOme \hSig-\bSig \bOme \bSig\|_{\infty}\\
&\leq & \|(\hSig-\bSig) \bOme (\hSig-\bSig) \|_{\infty}+2\|\bSig \bOme (\hSig-\bSig)\|_{\infty}.
\eeqrs
Under the conditions of the Proposition \ref{prop1}, 
\beqr \label{var}
\var\{E(Y\mid \x)\}=\bb \trans \bSig \bb+2 \tr (\bOme \bSig \bOme \bSig) \leq EY^2 < \infty.
\eeqr
We thus conclude  $\|\bOme\|_{\infty}< \infty$ and $\|\bOme \bSig\|<\infty$.
Then,
$\|(\hSig-\bSig) \bOme (\hSig-\bSig) \|_{\infty}\leq s_p \|\bOme\|_{\infty}  \|\hSig-\bSig\|^2_{\infty}
=o(1)  (n^{-1}\log p)^{1/2},
$and
$\pr\{ \|\bSig \bOme (\hSig-\bSig)\|_{\infty} \geq C \{\log(p)/{n}\}^{1/2}\} =O(p^{-1}) 
$ by invoking Lemma \ref{lem03} and the fact
$\|\bSig \bOme (\hSig-\bSig)\|_{\infty}=\max\limits_{i,j}|\e_i \trans \bSig \bOme (\hSig-\bSig) \e_j|.
$ Consequently, there exist a constant $C_2>0$ such that $\Delta_1\leq  C_2 \{n^{-\alpha/({\alpha+1})}\log{(p)}\}^{1/2}$ with probability larger than $1-O(p^{-1})$.
Set $\lambda_{1n}= {3}{\kappa}^{-1}C_2  \{n^{-\alpha/({\alpha+1})}\log{(p)}\}^{1/2}$ and 
by Lemma \ref{lem05}. We can conclude that with probability larger than $1-O(p^{-1})$, $\{\hOme_y\}_\supp=\textbf{0}$, and
$\|\hOme_y-\bOme\|_{\infty}\leq {4}{\kappa}^{-1}C_2 M  \{n^{-\alpha/({\alpha+1})}\log{(p)}\}^{1/2}.
$ The proof is now completed.  \hfill$\fbox{}$

\subsection{Appendix E: Proof of Theorem \ref{thm2}}
Given  $\wh \bb$,
\beqrs
\hLam_r 
&=&n^{-1}\sum_{i=1}^n\{ (Y_i-\bar{Y})-(\x_i-\bar{\x}) \trans \bb \}  (\x_i-\bar{\x})(\x_i-\bar{\x}) \trans\\
&&+ n^{-1}\sum_{i=1}^n (\x_i-\bar{\x}) \trans (\bb-\wh \bb)  \cdot (\x_i-\bar{\x})(\x_i-\bar{\x}) \trans \defby  \A_1+\A_2.
\eeqrs
Given true $\bb$,  \eqref{var} ensures that $\bb \bSig \bb \leq EY^2<\infty$, indicating that $\|\bb\|<C$ for some constant $C$.  Thus, $E \{\exp (c_1 |Y-\b \trans \x|^\alpha)\}\le C_1 < \infty$ and with probability greater than $1-O(p^{-1})$,
\beqr \label{thm2a1}
\|\A_1-\bLam\|_\infty \leq C_1  \{n^{-\alpha/(\alpha+1)}\log{(p)}\}^{1/2}.
\eeqr
Writing $ \wh \bb-\bb=(\eta_1,\cdots,\eta_p)\trans=\sum\limits_{k=1}^p \eta_k \e_{k},$
we have,
\beqrs
\Big\|\frac{1}{n}\sum_{i=1}^n (\wh \bb-\bb) \trans \x_i   (\x_i \x_i \trans)\Big\|_\infty &=&\Big\|\frac{1}{n}\sum_{i=1}^n \left(\sum_{k=1}^p \eta_t \e_{k} \right)\trans \x_i ( \x_i \x_i \trans)\Big\|_\infty \\
&\leq& \sum_{k=1}^p |\eta_k| \cdot \Big\|n^{-1}\sum_{i=1}^n \e_{k} \trans \x_i ( \x_i \x_i \trans)\Big\|_\infty,
\eeqrs
For $\e_k$,   $\E \{( \e_k \trans \x ) ( \x \x \trans)\}=0$.  By Lemma \ref{lem04}, there exists a large constant $C_2$ such that,
\beqrs
\pr\Big\{ \Big\|n^{-1}\sum_{i=1}^n \e_k \trans \x_i  (\x_i \x_i \trans)\Big\|_\infty\geq C_2 ( n^{-{2}/{3}}\log{(p)})^{1/2} \Big\}\leq p^{-2},
\eeqrs
which implies
\beqrs
&&\pr\Big\{\Big\|n^{-1}\sum_{i=1}^n (\wh \bb-\bb) \trans \x_i ( \x_i \x_i \trans)\Big\|_\infty \geq C_2 \sum_{k=1}^p |\eta_k|( n^{-{2}/{3}}\log{(p)})^{1/2} \Big\}\\
&\leq& \sum_{k=1}^p \pr\Big\{ \Big\|n^{-1}\sum_{i=1}^n \e_k \trans \x_i ( \x_i \x_i \trans)\Big\|_\infty\geq C_2 ( n^{-{2}/{3}}\log{(p)})^{1/2} \Big\}\leq p^{-1}.
\eeqrs
Note that $\sum\limits_{k=1}^p |\eta_k|=\|\wh \bb-\bb\|_1$. With probability greater than $1-p^{-1}$,
\beqrs
\Big\|n^{-1}\sum_{i=1}^n (\wh \bb-\bb) \trans \x_i ( \x_i \x_i) \trans\Big\|_\infty \leq C_2 \|\wh \bb-\bb\|_1 \{ n^{-2/3}\log{(p)}\}^{1/2},
\eeqrs
which together with Lemma \ref{lem03} yields 
\beqr \label{thm2a2}
\|\A_2\|_\infty \leq  C_3 \|\wh \bb-\bb\|_1 \{ n^{-2/3}\log{(p)}\}^{1/2} .
\eeqr
Combing \eqref{thm2a1} and \eqref{thm2a2}, with probability greater than $1-O(p^{-1})$,
\beqr
&& \|\hLam_r-\bLam\|_\infty \\
&\leq&\|\hLam_r-\bLam\|_\infty  C_4 \{n^{-\alpha/(\alpha+1)}\log{(p)}\}^{1/2}+C_5 \|\wh \bb-\bb\|_1 \{ n^{-2/3}\log{(p)}\}^{1/2} .\nonumber
\eeqr
Similarly to the proof of the Theorem \ref{thm1}, we can set
\beqrs
\lambda_{2n}=C_6 \{n^{-\alpha/(\alpha+1)}\log{(p)}\}^{1/2}+C_7 \|\wh \bb-\bb\|_1 \{ n^{-2/3}\log{(p)}\}^{1/2} 
\eeqrs
and conclude that with probability lager than $1-O(p^{-1})$, $\{\hOme_r\}_\supp=\textbf{0}$ and
$
\|\hOme_r-\bOme\|_{\infty}\leq C_8 M\lambda_{2n},
$
for some constant $C_8$.
\bibliography{ref}
\end{document}